\documentclass[pra,aps,superscriptaddress,twocolumn,nopacs,nofootinbib]{revtex4}
\usepackage{graphicx,color,epstopdf}
\usepackage{amsmath,amsfonts,enumerate,amsthm,amssymb,mathtools}
\usepackage{enumitem}
\usepackage{thmtools,thm-restate}
\usepackage{bbold}
\usepackage{hyperref}
\usepackage{multirow}
\usepackage{subfigure} 
\usepackage{mathdots} 
\usepackage{soul}

\usepackage{eufrak}
\usepackage{array}

\hypersetup{
colorlinks=true, 
linkcolor=blue, 
citecolor=blue 
}
\usepackage{amsmath}
\usepackage{amsfonts}
\usepackage{amssymb}
\usepackage{graphicx}
\usepackage{color}
\usepackage[latin1]{inputenc}
\usepackage{wrapfig}

\usepackage{amsthm}

\newtheorem{theorem}{Theorem}
\newtheorem{thm}{Theorem}

\newtheorem{lemma}[thm]{Lemma}

\newtheorem{dfn}{Definition}

\newtheorem{prop}[thm]{Proposition}
\newtheorem{obs}[thm]{Observation}

\newtheorem{cor}[thm]{Corollary}

\def\id{I}

\newcommand{\be}{\begin{equation}}
\newcommand{\ee}{\end{equation}}

\newsavebox{\A}
\savebox{\A}{$\left(\begin{smallmatrix}
1-q_{10}-q_{20}&1&1\\
q_{10}&0&0\\
q_{20}&0&0\\
\end{smallmatrix}\right)$}

\newsavebox{\G}
\savebox{\G}{$\left(\begin{smallmatrix}
1&0&0\\
0&1-q_{21}&1\\
0&q_{21}&0\\
\end{smallmatrix}\right)$}

\newsavebox{\B}
\savebox{\B}{$\left(\begin{smallmatrix}
1-q_{10}&1-q_{21}&1\\
q_{10}&0&0\\
0&q_{21}&0\\
\end{smallmatrix}\right)$}

\newsavebox{\Pu}
\savebox{\Pu}{$\left(\begin{smallmatrix}
1-q_{20}&0&1\\
0&1&0\\
q_{20}&0&0\\
\end{smallmatrix}\right)$}

\newsavebox{\smlmat}
\savebox{\smlmat}{$\left(\begin{smallmatrix}
1&3&4\\
2&0&0\\
5&0&0\\
\end{smallmatrix}\right)$}

\newsavebox{\smmat}
\savebox{\smmat}{$\left(\begin{smallmatrix}
0&0&2\\
0&7&0\\
5&0&1\\
\end{smallmatrix}\right)$}

\newsavebox{\mat}
\savebox{\mat}{$\left(\begin{smallmatrix}
4&0&0&0\\
5&3&6&7&\\
0&0&2&0\\
0&0&0&1\\
\end{smallmatrix}\right)$}

\newsavebox{\matP}
\savebox{\matP}{$\left(\begin{smallmatrix}
1&1&0&0\\
1&0&1&0&\\
1&0&0&1\\
1&0&0&0\\
\end{smallmatrix}\right)$}

\def\diam{\rm{diam}\,}

\begin{document}
\title{Generalized XOR non-locality games with graph description on a square lattice}

\author{Monika Rosicka}
\affiliation{Institute of Theoretical Physics and Astrophysics, National Quantum Information Centre, Faculty of Mathematics, Physics and Informatics, University of Gda\'nsk, 80-308 Gda\'nsk, Poland}
\author{Pawe\l{} Mazurek}
\affiliation{Institute of Theoretical Physics and Astrophysics, National Quantum Information Centre, Faculty of Mathematics, Physics and Informatics, University of Gda\'nsk, 80-308 Gda\'nsk, Poland}
\author{Andrzej Grudka}
\affiliation{Faculty of Physics, Adam Mickiewicz University, 61-614 Pozna\'n, Poland}
\author{Micha{\l} Horodecki}
\affiliation{Institute of Theoretical Physics and Astrophysics, National Quantum Information Centre, Faculty of Mathematics, Physics and Informatics, University of Gda\'nsk, 80-308 Gda\'nsk, Poland}
\affiliation{International Centre for Theory of Quantum Technologies, University of Gda\'nsk, 80-308 Gda\'nsk, Poland}

\begin{abstract}
We propose a family of non-locality unique games for 2 parties based on a square lattice on an arbitrary surface. We show that, due to structural similarities with error correction codes of Kitaev for fault tolerant quantum computation, the games have classical values computable in polynomial time for $d=2$ measurement outcomes. By representing games in their graph form, for arbitrary $d$ and underlying surface we provide their classification into equivalence classes with respect to relabeling of measurement outcomes, for a selected set of permutations which define the winning conditions. A case study of games with periodic boundary conditions is presented in order to verify their impact on classical and quantum values of the family of games. It suggests that quantum values suffer independently from presence of different winning conditions that can be imposed due to periodicity, as long as no local restrictions are in place.
\end{abstract}

\maketitle

\section{Introduction}
The fact that physical theories simultaneously aim at explaining phenomena already observed, and need to be, at least in principle, experimentally falsifiable, puts measurement process in the centre of attention. Statistics of single measurement outcomes, and the nature of correlations between them, depends on the physical theory describing the process of generating measurement outcomes.  

It cannot be a surprise that access to systems exhibiting richer statistics of measurement outcomes enables one to improve performance in different tasks related to information processing. For example, the non-existence of a hidden-variable model for observed measurement outcomes, attested by violation CHSH inequality, is a sufficient condition for a string of outcomes to be intrinsically random, i.e. to exhibit randomness that cannot be explained by our lack of knowledge about the system \cite{Pironio2010}. Intrinsic randomness of quantum correlations can be used for secure establishment of a cryptographic key \cite{Bennett1984}, even in a situation when one relies solely on the measurement statistics, without a need to trust that a system has been prepared, and measurement performed in a specific way \cite{Acin2007}. Furthermore, a gap of effectiveness for usage of quantum and classical resources is present in communication complexity tasks \cite{Buhrman2010}.        

The qualitative difference between classical and quantum systems, through the notion of Bell non-locality \cite{Bell1964}, can be quantified in two equivalent frameworks: Bell inequalities and non-locality games. In the game setting, the correlators present in the Bell inequality may be assigned some desired values, and the task for the physical system may be set to satisfy all the given constraints with the highest probability, with respect to a previously known probability distribution of measurement settings. This is interpreted as questions to parties of the physical system. There exist games with optimal classical and quantum strategies proven to yield different values of winning probability. The most famous example is the game associated with the CHSH Bell inequality. 

Taking this into account, it is of vital importance to calculate quantum and classical values achievable for a given non-locality game. In this paper, we propose a family of non-locality unique games for two parties, defined on a square lattice on an arbitrary surface. The games show a lot of similarities with error correction codes proposed for fault tolerant quantum computation by Kiteav \cite{Kitaev2003}. We use them to calculate their classical values in a polynomial time (while it is an NP-hard problem in general), for number of measurement outcomes $d=2$. Furthermore, due to these geometrical properties we are also able to establish a classification of games into equivalence classes with respect to local relabeling of outcomes of measurements, and to study the role which periodic boundary conditions can play in setting of classical and quantum values of these games. We assume that the periodic lattice has even number of cells in order to allow for bipartite structure needed for non-locality game; otherwise, the construction can be used to define a contextuality game. We also point out important differences between Kitaev codes and our games. 

The above is achieved through representation of non-locality games in a graph form. The graph description of non-locality games, its basic properties, as well as general notation and calculation of games classical and quantum values, are introduced in Section \ref{Pre}. Sections \ref{2}, \ref{3}, \ref{sub:4+} are devoted to classification of games for gradually incising number of possible measurement outcomes, provided that the winning conditions are described by a specific group of permutations.
Section \ref{6} contains a description of the procedure for calculation of classical values for $d=2$, and discussion of possible extensions to higher dimensions, one being a generalization of the method for $d=2$, while the second one based on an unique representation of each game with respect to a chosen maximal spanning tree. Section \ref{ex} is devoted to analytic and numerical studies of the role which periodic boundary conditions can play in setting classical and quantum values for the family of games. We conclude with Section \ref{con}.

\section{Preliminaries}\label{Pre}

\subsection{Non-locality games for 2 parties}

The setting of the game is the following: a referee asks a question $x$ from the set $A=\{A_{1},\dots,A_{|A|}\}$ to one part of the spatially separated system, Alice, and a question $y\in B=\{1,\dots,B_{|B|}\}$ to another party, Bob. Then Alice and Bob return answers $a\in\{0,\dots,d-1\}$ and $b\in\{0,\dots,d-1\}$, respectively. Case of different number of possible answers for Alice and Bob can also be described in this way, with a subset of answers not used by one party. The parties cannot communicate after receiving the questions, and they return answers such that the probability of winning the game is maximized, with uniform probability of different pairs of questions to appear. The winning conditions (i.e., a set of accepted answers for given questions), are known to both parties. The choice of answers can be based on a pre-established classical strategy (then the maximal winning probability is called a classical value of the game ) or outcomes of measurements on a shared quantum state (with maximal winning probability called quantum value of the game).   

The dimension of underlying quantum systems and measurement operators can be in principle arbitrary. This makes it calculating the quantum value difficult. It is still not known if this is possible for an arbitrary game, although semidefinite programs can be used to compute upper bounds \cite{Navascues2015}. Therefore, one is interested in a special class of non-locality games, called \textit{unique games}. They are defined by the property that for every pair of questions and an answer by Alice: $(x,y,a)$, there exists exactly one answer $b$ that satisfies the winning condition of the game. In other words, for each pair of questions the constraint which the winning answers must satisfy is defined by pre-established function $\pi:A_A\mapsto A_B. $ The players win iff $b=\pi(a)$. 
For this class of games, it has been showed that the quantum value can be approximated to a constant factor in polynomial time \cite{Kempe2010}. Furthermore, in the scenario of $XOR$ game, where $A_{A}=A_{B}=2$ and winning conditions depend solely on $a\oplus b$ for a given $(x,y)$, the quantum value can be computed exactly in polynomial time \cite{Cleve2004, Wehner2010} due to Tsirelson theorem \cite{Tsirelson1980}. 

On the other hand, calculating the classical value of a non-locality game is a vertex labeling problem, and as such, for every positive constant $\delta$, there is always $A_{A}=A_{B}$ high enough such that it is NP-hard even to decide whether a given game has a classical strategy satisfying all winning conditions, or there is no strategy satisfying more than $\delta$ fraction of the winning conditions \cite{Arora1998,Arora1998B,Raz1998}. If the Unique Games Conjecture \cite{Khot2002} is true, then the above  applies as well to unique non-locality games,  with the modification that the task is to distinguish between existence of a strategy satisfying almost all winning conditions, and a non-existence of strategy satisfying more than a small fraction of winning conditions. For unique games in general, and XOR games in particular, it is known that calculating their exact classical values is an NP-hard problem \cite{Hastad2001}.    

In this paper we propose a class of two party non-locality unique games for arbitrary local dimensions $d=A_{A}=A_{B}$. This class is defined on bipartite graphs, with vertices of the graph corresponding to questions asked by the referee, and edges labeled by permutations between measurement outcomes that satisfy the winning condition.

\subsection{Graph description}

\label{sec:labeledgraphs}
Below we introduce basic notions associated with graph representation of non-locality games. In this representation, the vertices of a graph correspond to questions asked by the referee (measurements). Two vertices are connected by an edge iff both corresponding measurements can be performed simultaneously. A function $K:E(G)\mapsto S_d$ assigns to each edge a permutation of the set of $d$ elements. These permutations represent the desired correlations between measurement outcomes.
If the graph is connected and bipartite, then each of the two independent sets of vertices corresponds to measurements performed by a distinct party. Then the labeled graph directly corresponds to a non-locality game. The classification provided in further chapters applies to generalized XOR games, which have been investigated in detail in \cite{Rosicka2016}. XOR games are characterized by binary measurement outcomes ($A_{A}=A_{B}=2$) with correlations/anticorrelations demanded between outcomes of selected parties. In the generalized version of a XOR game,  we allow for $A_{A}=A_{B}=d$ possible outcomes from the set $\{0, ...,d-1\}$.

 Constraints on the edges connecting vertices $u$ and $v$ are defined by permutations $S_{n}$ of the set $\{0, ..., n-1\}$. We will focus on games in which these permutations belong to the set $L_n=\{\tilde{\pi_i}:\tilde{\pi_i}(x)=i-x\mod n\}$ or $L_n'=\{\tilde{\sigma_i}:\tilde{\sigma_i}(x)=x+i \mod n\}.$ A graph together with an edge labeling will be referred to as a \textit{labeled} graph. If all permutations assigned to the edges are equal to their inverse, we will talk about an undirected labeled graph. In this paper, we will be interested in connected bipartite graphs defined on a cubit lattice, but many of the results presented here do not depend on the type of the connected bipartite graph.  
     
\begin{figure}[h!]
\includegraphics[scale=0.25]{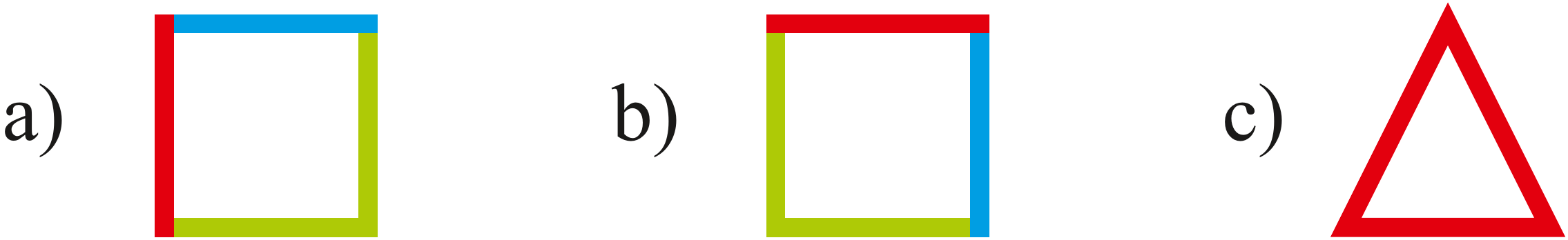}
\caption{\label{fig:1} Examples of cycles for $d=3$: a) good, b) ugly, c) bad. Colors correspond to permutations from the group $L_3=\{\tilde{\sigma}_i:\tilde{\sigma}_i(x)=i-x \mod 3\}$, with $\tilde{\sigma}_0$ (red) preserving value 0, $\tilde{\sigma}_1$ (blue) preserving value 2, $\tilde{\sigma}_2$ (green) preserving value 1. $x$ is the input of the permutation, whereas $i$ labels permutations.}
\end{figure}
				
We will start with characterization of  amount of classicality for a XOR game $d=2$. For this we will use a notion of consistency. An assignment $k:V(G)\mapsto \{0,\dots,d-1\}$ of measurement outcome values to vertices is \textit{consistent} if it has no contradiction on any edge of the graph, i.e. for every edge $uv$ the relation between outcomes on its vertices is given by a permutation labeling this edge, $k(v)=\pi(k(u))$. A connected labeled graph can have no more than $d$ consistent assignments, as assigning a value to one vertex determines the values of all its neighbors. 
We will say that a labeled graph (or its subgraph) is \textit{good} if it has $d$ consistent vertex-assignments and \textit{bad} if no assignment is consistent. If the number of consistent assignments is larger than $0$ but less than $d$, we say that the graph is \textit{ugly} (see Fig. \ref{fig:1}). 

Every consistent assignment defines a deterministic strategy for a given game, which allows the players to win with probability 1. If a labeled graph has no consistent assignment, then no such strategy exists for the game. Thus, good and ugly graphs will describe games that can be won by strategies in which a state of the system is purely classical, and proper measurements just reveal the properly correlated values, whereas in games represented by bad graphs one can expect that quantum strategies may outperform classical ones.

\subsection{Equivalence of labeled graphs}

The notion of \textit{equivalence} between two games has to be properly defined in the language of their graph representation. \textit{We will say that two games are equivalent iff the corresponding labeled graphs are equivalent.} We say that two labeled graphs are equivalent iff one can be obtained from the other through:
\begin{enumerate}
\item an isomorphism of the underlying graphs 
\item changing the direction of an edge and replacing the permutation on this edge with its inverse
\item switching operations $s(v,\sigma)$, which changes the labels on all edges incident with the vertex $v$ as follows:
\begin{enumerate}
\item if $\overrightarrow{uv}\in E(G),$ we replace $K(\overrightarrow{uv})=\pi$ with $K'(\overrightarrow{uv})=\sigma\pi$,
\item if $\overrightarrow{vu}\in E(G),$ we replace $K(\overrightarrow{vu})=\pi$ with $K'(\overrightarrow{vu})=\pi\sigma^{-1}$.
\end{enumerate}
\end{enumerate}

Each of the above operations can be interpreted as renaming the inputs and/or outputs. It follows that equivalent games have equal classical and quantum winning probabilities.

In this paper, however, we will largely focus on the equivalence between different labelings on the same graph. We say that two labelings of a graph are \textit{equivalent} iff one can be obtained from the other through switches.
It is clear that any two games defined on the same graph with equivalent labelings must be equivalent.

\begin{figure}[h!]
\includegraphics[scale=0.25]{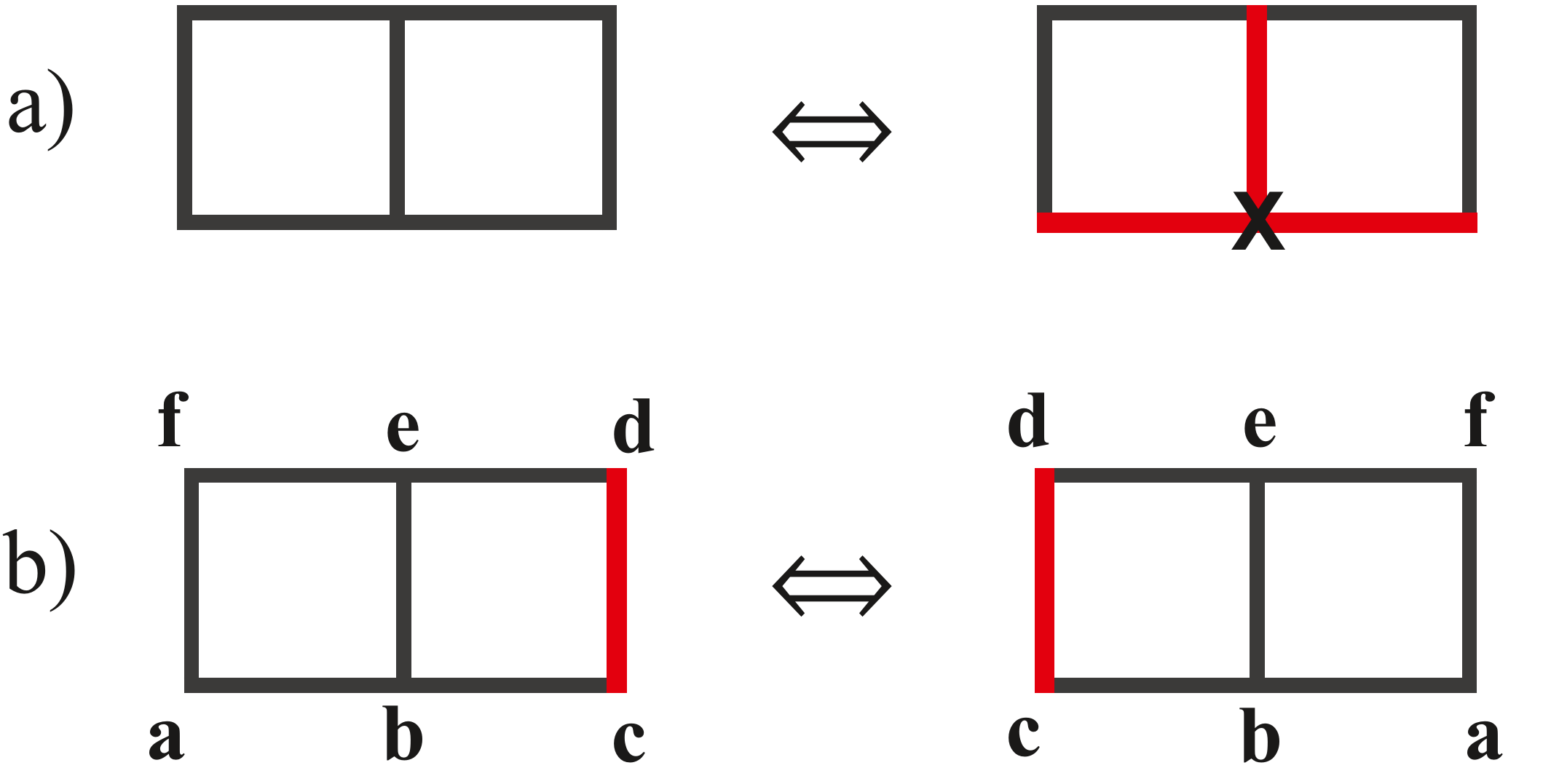}
\caption{\label{fig:2}Examples of equivalent games for $d=2$: a) equivalence in terms of switches (relabeling measurement outputs), b) equivalence in terms of graph isomorphism (relabeling measurements). Black and red edges represent permutations $\sigma_0=\id$ and $\sigma_1=(01)$, respectively. 
$X$ denotes vertices where switches $\tilde{\sigma}_1$ are applied. 
}
\end{figure}

 For games represented by a labeled graph on a planar grid, their classification will depend only on a \textit{local} structure of the graph. Speaking more precisely, the equivalence of two labelings of such a graph will be determined by sets and type of cycles defined on cells of the grid. By a cell we mean here a cycle that does not contain any other cycles (In a square lattice, this is a cycle with four edges). A bad cell will be referred to as a \textit{defect}. 

For grids on surfaces other than the plane (eg. a torus), in order to classify corresponding non-locality games we will have to take into account cycles arising from the topological structure. This is similar to the way in which classes of homology of error paths have to be taken into account in topological error correction codes in order to describe a logical state of a code, and we will comment on observed similarities and differences between the two.

\section{Classification of non-locality games for $d=2$}\label{2}
The group of permutations of two elements does not have any non-trivial subgroups and consists only of identity and transposition operations: $S_{2}=\{Id,(01)\}$. This group is the example of a permutation group defined by $L_d'$, that, along the group $L_d$, will be subjected to a more detailed analysis for $d>2$ in the following chapters. Proofs of theorems will be based on a concept of a canonical representation of a graph, used in \cite{Zaslavsky1982} to prove equivalent statements for signed graphs (which are functionally identical to labeled graphs with $d=2$ outcomes). Later, we generalize this line of reasoning to the case of higher $d$. 

\begin{figure}[h!]
\includegraphics[scale=0.25]{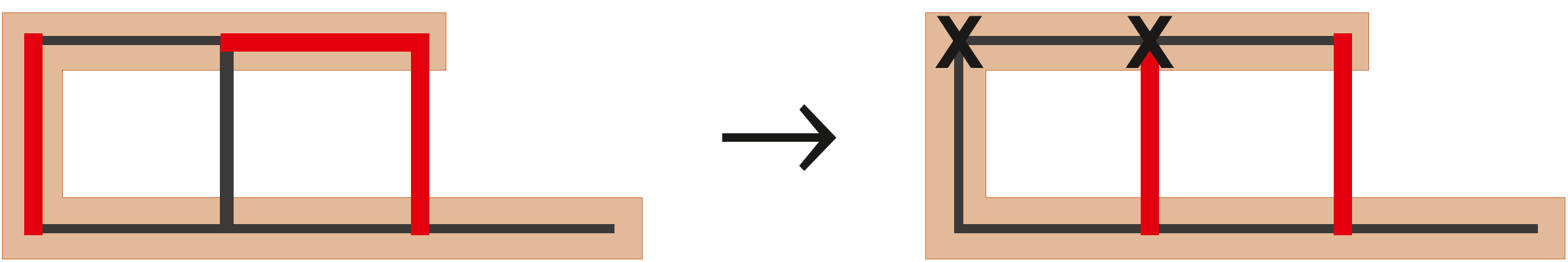}
\caption{\label{figura3}Transition from a labeled graph (left) to its canonical form (right) by switches ($\tilde{\sigma}_{1}\in L_{2}$) applied to vertices marked with 'X'. A selected spanning tree encircled in brown.}
\end{figure}

A \textit{spanning tree} of a graph is a subgraph containing all of its vertices and some of the edges such that there is exactly one path connecting each two vertices of the graph. We use this to define the \textit{canonical representation} of a game. The canonical representation of a game $(G,K)$ with respect to the spanning tree $T$ is a game on the same graph with a labeling equivalent to $K$ such that the $\id$ permutation is assigned to all edges of $T$.
This is similar to a concept introduced in \cite{Zaslavsky1982} for signed graphs. Later, we shall use a generalized version for graphs labeled with $S_d$ for an arbitrary $d$.
It is clear from the definitions that one can always obtain a canonical representation of any game through switches (see Fig. \ref{figura3}). 
For $d=2$, the equivalence classes of graphs are uniquely determined by their canonical representation: 

\begin{prop}\label{Th_Equiv}
Two labelings of the same graph with $d=2$ outcomes are equivalent iff the corresponding games have the same canonical representation. The canonical representation can be defined with respect to any spanning tree.  
\end{prop}

\begin{figure}[h!]
\includegraphics[scale=0.2]{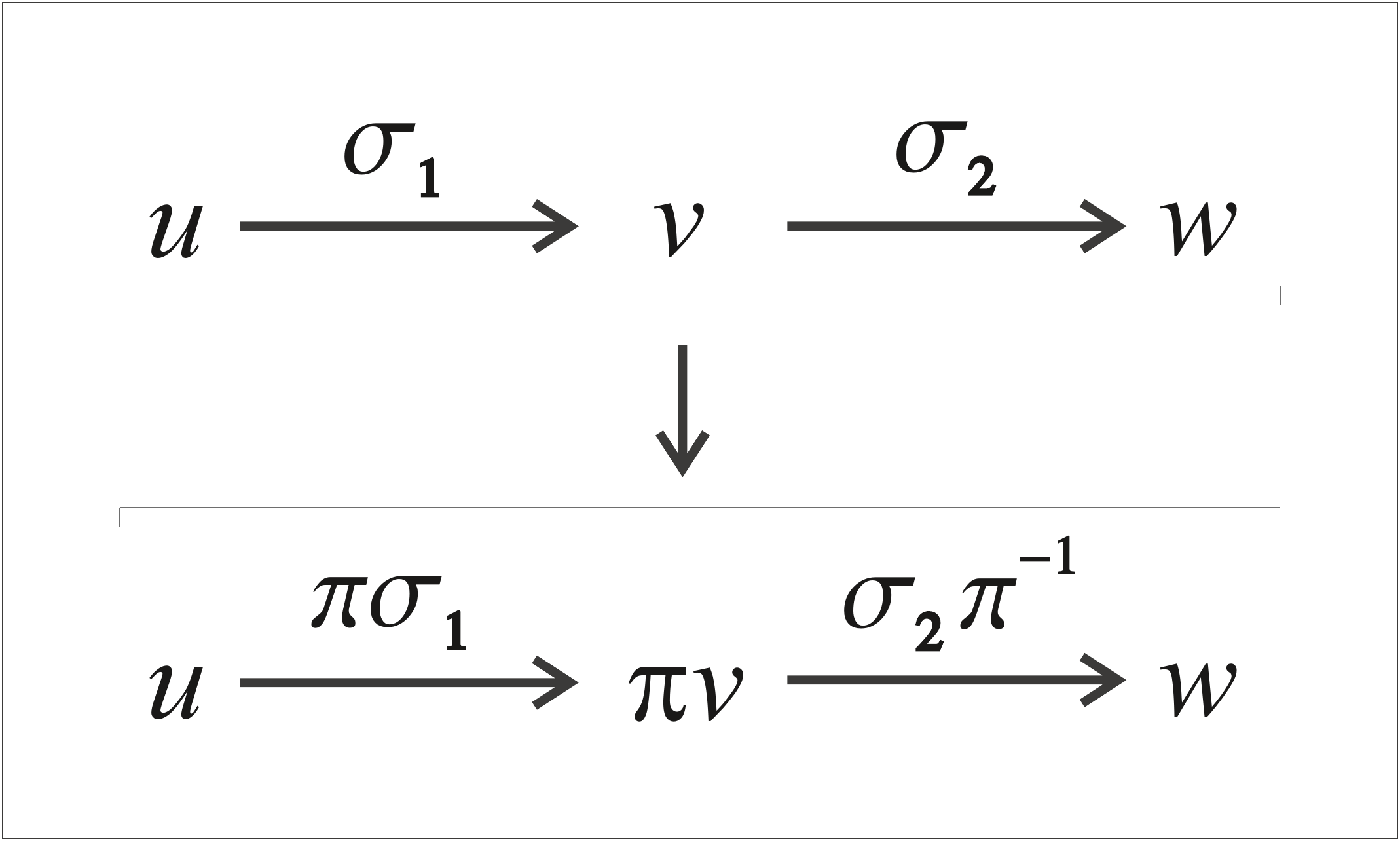}
\caption{\label{figura4}A transformation of permutations due to application of switch $s(v,\pi)$.}
\end{figure}

Before we move to the proof, let us stress that Proposition {\ref{Th_Equiv}} is valid for games defined on surfaces with and without boundary conditions, as in both these cases it is possible to transform a game to its canonical form.
\begin{proof}
($\Leftarrow$ part).
The same canonical representation of two games implies that one can be transformed into the other by performing switches that bring one of them to the common canonical form, and then transform the canonical form into the other game. Therefore, the games are equivalent. 

($\Rightarrow$ part). It follows from the fact that, for $d=2$, a game has only one canonical representation with respect to a given spanning tree. To see this, let us notice that $S_{2}=L_{2}'=L_{2}$, and show how the structure of labelings changes due to local permutations. Let $\overrightarrow{uv}$ be an edge originally labeled with a permutation $\sigma$ ($\sigma(u)=v$, where we abuse the notation and denote by $u, v$ labelings of the vertices), and we apply a switch, then $\sigma$ changes to $\pi\sigma$ for the switch $s(v,\pi)$, and to $\sigma\pi^{-1}$ for $s(u,\pi)$ (see Fig. \ref{figura4}). 
According to this rule, any switch applied on one vertex belonging to an edge outside the spanning tree and aimed at changing the permutation assigned to this edge, would have to be accompanied by an inverse transformation on neighboring vertices that are connected through the spanning tree. But the inverse of $(01)$ is $(01)$. Therefore, we have to apply the same permutation on every other vertex of the spanning tree when we construct the new canonical representation of the game, so that the $Id$ permutations on the spanning tree remain unaffected. But this would imply performing a switch on both ends of every edge not in the spanning tree, and as permutations belonging to $S_{2}$ commute, we have $(01)\sigma(01)^{-1}=\sigma$, and permutations $\sigma\in\{Id, (01)\}$ assigned to all such edges remain unchanged. Thus the canonical representation will remain the same. 

It follows that two labelings of a graph with $S_2$ are equivalent if and only if they have a shared canonical representation with respect to an arbitrary spanning tree.
\end{proof}

Notice that the first part of the proof does not depend on the number of outputs. Hence we have the following result.

\begin{cor}\label{Cor_1}
If two games for any $d$ have the same canonical representations (on an arbitrarily selected spanning tree), then they are equivalent.
\end{cor}

For $d=2$, the following holds as well 
\begin{theorem}\label{Th1}
Two labelings of a graph with $S_2$ are equivalent iff they have the same set of bad cycles. 
\end{theorem}

Note that the above is the definition of equivalence for signed graphs in \cite{Harary1953}. In the simple case of $d=2$, every bad cycle contains odd number of transpositions.

\begin{proof}
($\Rightarrow$ part). 
For an arbitrary cycle, let $\sigma$ be the composition of all permutations along the cycle. If we perform a switch on an arbitrary vertex of the cycle, then $\sigma=\sigma_1\sigma_2$ becomes $\sigma_1\pi^{-1}\pi\sigma_2$ (or $\pi\sigma\pi^{-1}$ if the switch was on the starting vertex, but the permutations in $S_2$ commute, so $\pi\sigma\pi^{-1}=\sigma_1\pi^{-1}\pi\sigma_2$). Since $\sigma_1\pi^{-1}\pi\sigma_2=\sigma_1\sigma_2=\sigma$, no switch can change the permutation $\sigma$. Thus any two equivalent labelings have the same set of bad cycles.

($\Leftarrow$ part). It follows from Corollary \ref{Cor_1} that if two labelings are not equivalent, then their canonical representations with respect to the same spanning tree are different. Different canonical representations imply different sets of bad cycles, because one can always find a differentiating cycle. Let $e$ be an edge which differs between the two canonical representations. Any cycle consisting of $e$ and some edges of the spanning tree is good in one of the games and bad in the other (see Fig. \ref{figura5}). 
\end{proof}

\begin{figure}[h!]
\includegraphics[scale=0.25]{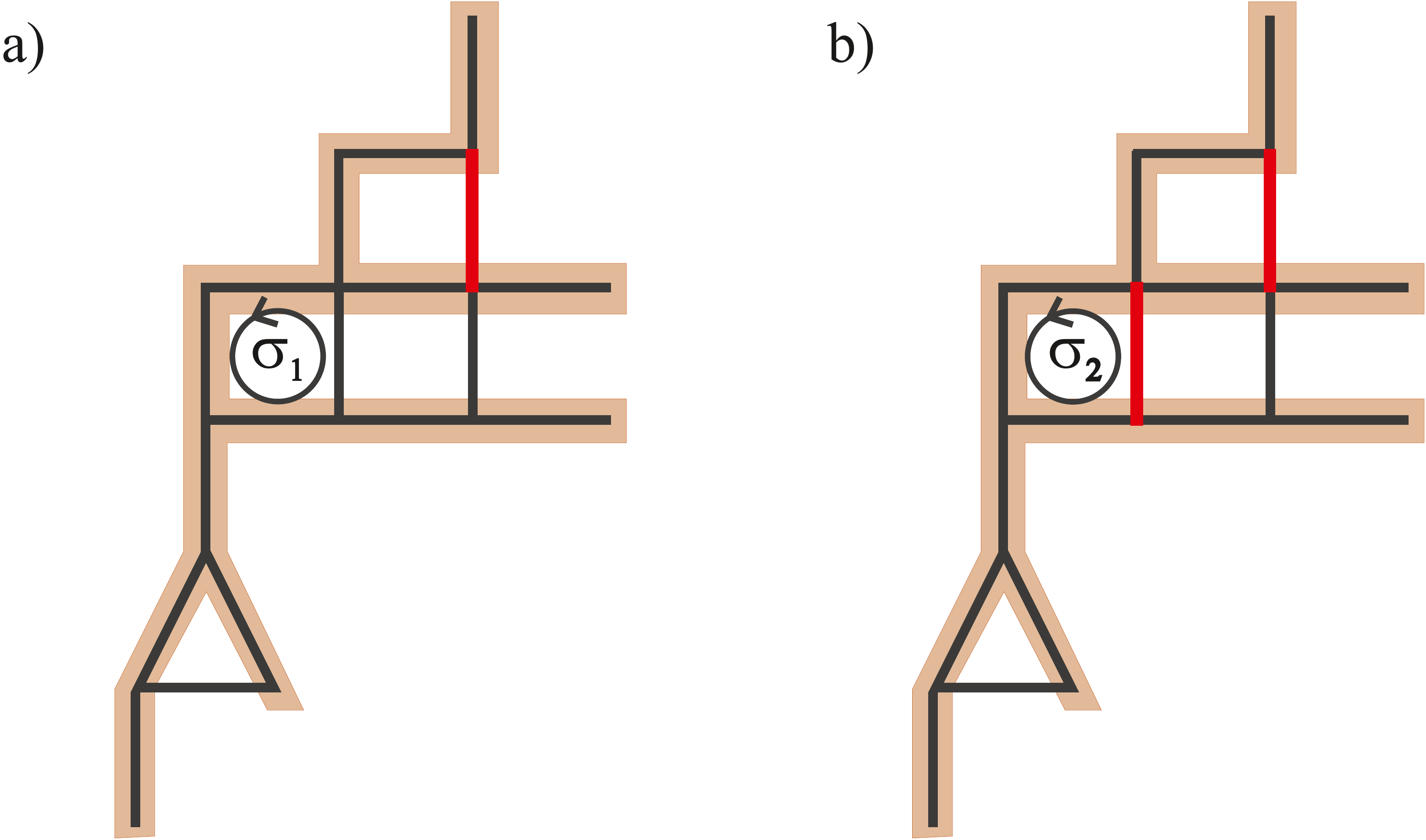}
\caption{\label{figura5}games with different canonical representation for $d=2$, and difference in set of bad cycles: $\sigma_{1}\neq \sigma_{2}$.}
\end{figure}

Bad cycles in a grid can have two origins -- they can arise as a result of the existence of defects, or can have a non-local character, i.e. are not a function of defects in the graph. These two origins have distinctive implications in the classification of games, therefore we will present this classification separately for planar games and games with on surfaces other than the plane.

\subsection{Planar graphs for $d=2$}
We will show that for a planar graph for $d=2$ all bad cycles arise from defects.
Let us note that for $d=2$ the cycles can have only two classes: good and bad. A good cycle is a cycle with defect class $\id$ (or $0$) and a bad cycle is a cycle with defect class $(01)$ (or $1$).

\begin{figure}[h!]
\includegraphics[scale=0.45]{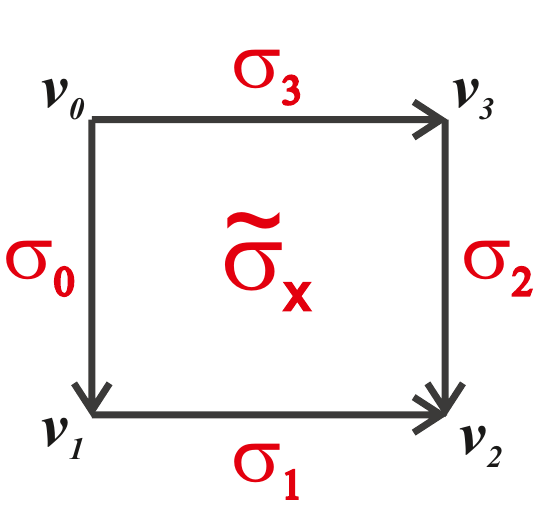}
\caption{\label{figura6} Defect class of a cell is associated with value $x$ of $\tilde{\sigma}_{x}=\sigma_3^{-1}\sigma_2^{-1}\sigma_1\sigma_0$, $\tilde{\sigma}_{x}\in L'_{3}$.}
\end{figure}


Proposition \ref{Tdodawanie_n2} provides an easy way of finding the defect class of a larger cycle based on the classes of the cells contained within. For a labeling with $d=2$ outcomes, the class of every cell is either $0$ or $1$, which implies the following.


\begin{cor}\label{cor_bad} 
For a planar graph labeled with $S_2$, the cycle is bad if it contains an even number of bad cells, otherwise it is good.
\end{cor}

By Theorem \ref{Th1} we know that for $d=2$ two labelings are equivalent iff they have the same set of bad cycles. From Proposition \ref{Tdodawanie_n2} we see that for a planar graph the set of bad cycles is uniquely associated with the set of defects. Therefore, we have

\begin{cor}\label{cor_2}
Two labelings of a planar graph with $S_2$ are equivalent iff they have the same set of defects.
\end{cor}

\begin{figure}[h!]
\includegraphics[scale=0.2]{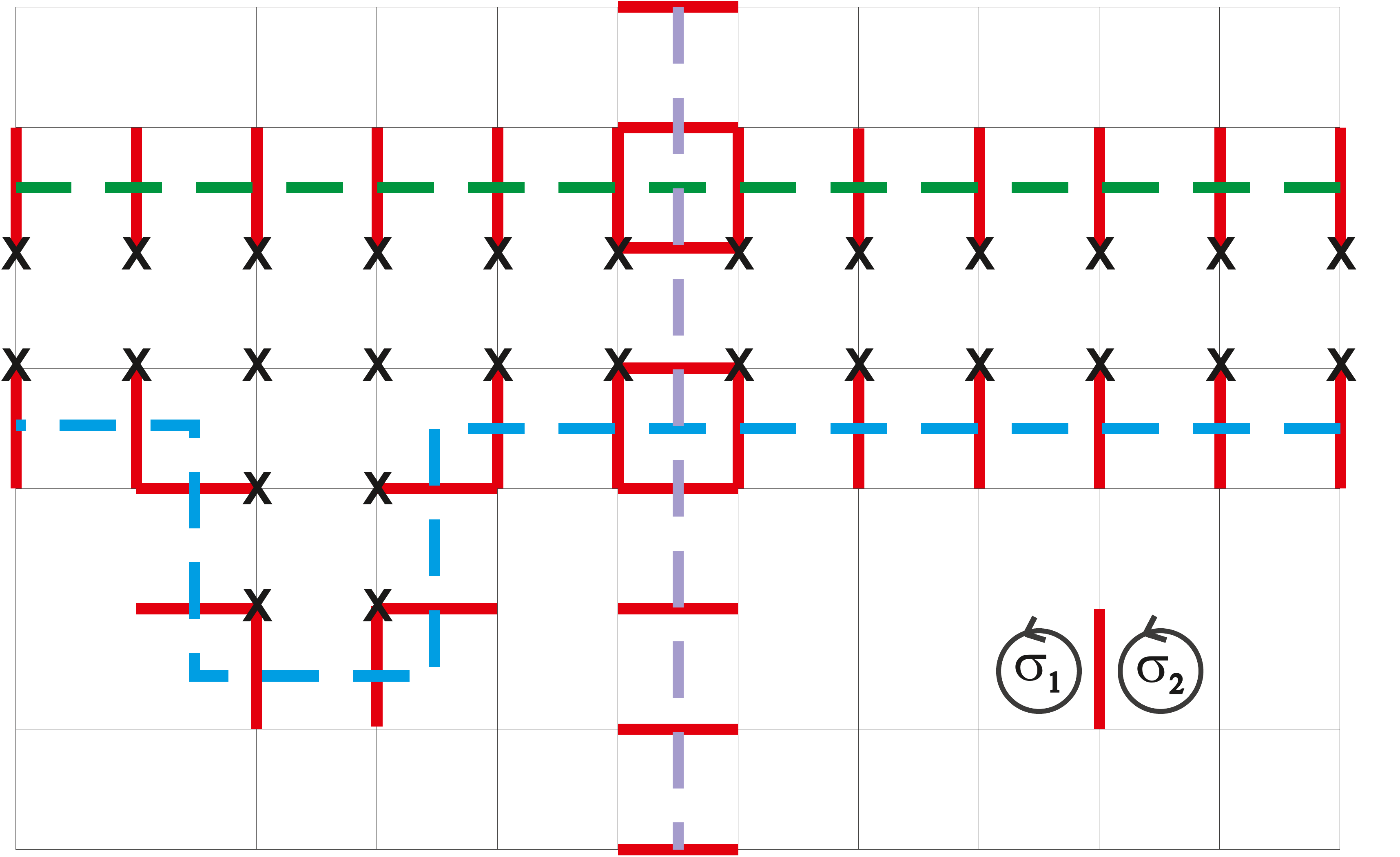}
\caption{\label{figura8}
A network with cyclic boundary conditions on left/right and bottom/up boundaries, and a graph for $d=2$. $\sigma_{1}$, $\sigma_{2}$ -- bad 
 cycles stemming from defects. Bad cycles of red permutations along green and blue lines -- two 
 cycles characterizing the same equivalence class of the game (transition between the cycles can be performed through $\tilde{\sigma}_{1}$ switches applied to vertices denoted by 'X'). A bad 
 cycle along the purple line, characterizing different equivalence class.}
\end{figure}

\subsection{Graphs with periodic boundary conditions for $d=2$}\label{periodic2}
If we admit for periodic boundary conditions on the grid, i.e. place the grid on some surface other than the plane, then bad cycles can arise due to lines of $(01)$ permutations in the dual lattice, that give rise to no defects (see Fig. \ref{figura8}). Fig. \ref{figura9} shows how the periodicity of boundary conditions creates opportunity for non-local paths of errors to arise -- they constitute paths from the center to the exterior of the continuously deformed dual graph.  

The effect of taking into account periodic boundary conditions is depicted in Fig. \ref{figura10}. If there are no boundary conditions, the edges of the graph can be divided into a spanning tree (brown) and a remaining set of black edges. Corollary \ref{cor_2} states that two games on a plane for $d=2$ are equivalent, iff they have the same set of defects. If the graph is driven to a canonical representation with respect to the spanning tree depicted, defects are determined only by the black edges. 

However, introducing a periodic boundary condition (in one or more possible directions) implies an addition of a column/row of additional edges. Naturally, labelings of these new edges will potentially give rise to new defects. Furthermore, as these new edges bypass the spanning tree of the graph, they can all be labeled by (01) permutations without creating a single defect. Such chains of permutations cannot be removed or contracted to a point by switches. We will refer to them as \textit{loops}.

\begin{figure}[h!]
\includegraphics[scale=0.18]{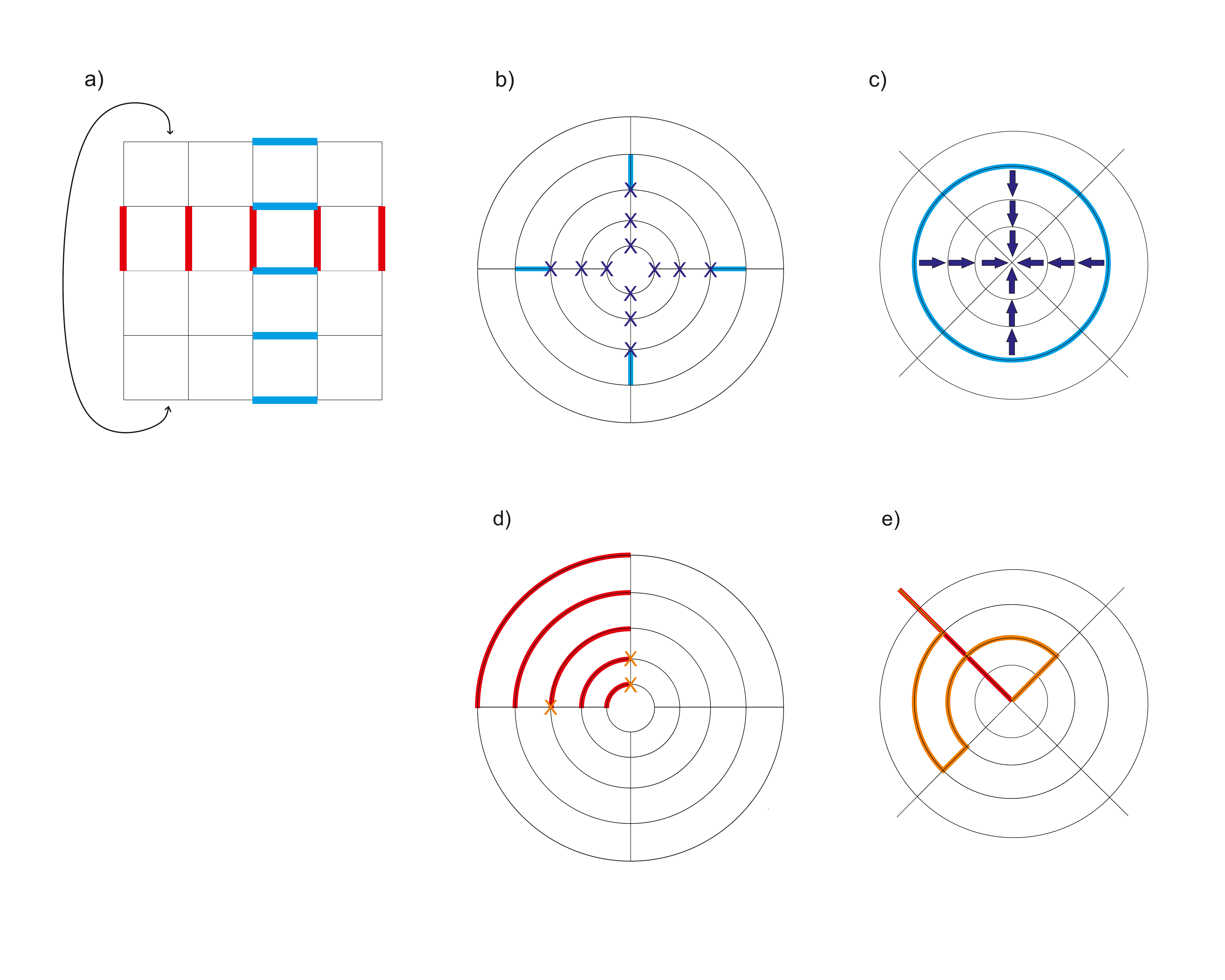}
\caption{\label{figura9}
Periodicity of boundary conditions (a) enables continuous transformation of a game to a form (b), (d). Both red and blue sets of (01) permutations lead to no defects. While blue permutations can be erased by applying switches on vertices labeled with $X$ (b), which contracts the associated path in a dual lattice to a point (c), red permutations remain joining the center of the dual lattice with its exterior (e) -- applying a selected choice of switches (d) gives rise to an orange line (e), with the same end points as the initial, red line.  
}
\end{figure}

\begin{figure}[h!]
\includegraphics[scale=0.35]{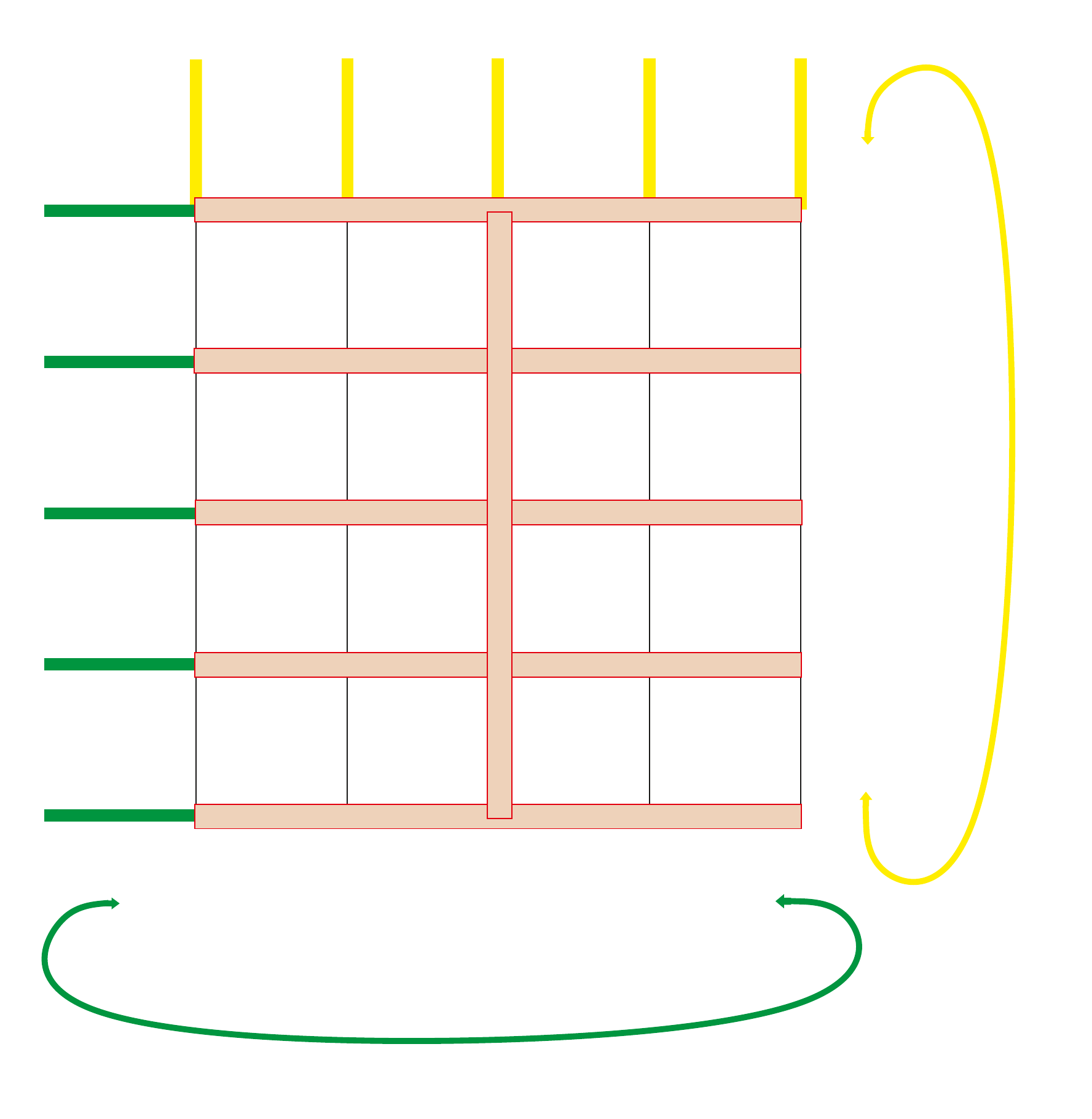}
\caption{\label{figura10}
Due to boundary conditions applied along directions marked by yellow/green arrows, the graph structure can be divided into the spanning tree 
$T$ (brown), black edges within the extended tree $T^L$ and yellow/green edges joining the boundaries and bypassing the spanning tree. Cells with black and brown edges contribute to the equivalence class of the graph by experiencing defects. Defects can be obtained on cells with yellow/green edges as well, but furthermore yellow/green edges can be collectively labeled by (01) permutations without creating a defect, despite the fact that such a labeling is not equivalent to labeling all edges with $\id$. }
\end{figure}

\begin{figure}[h!]
\includegraphics[scale=0.3]{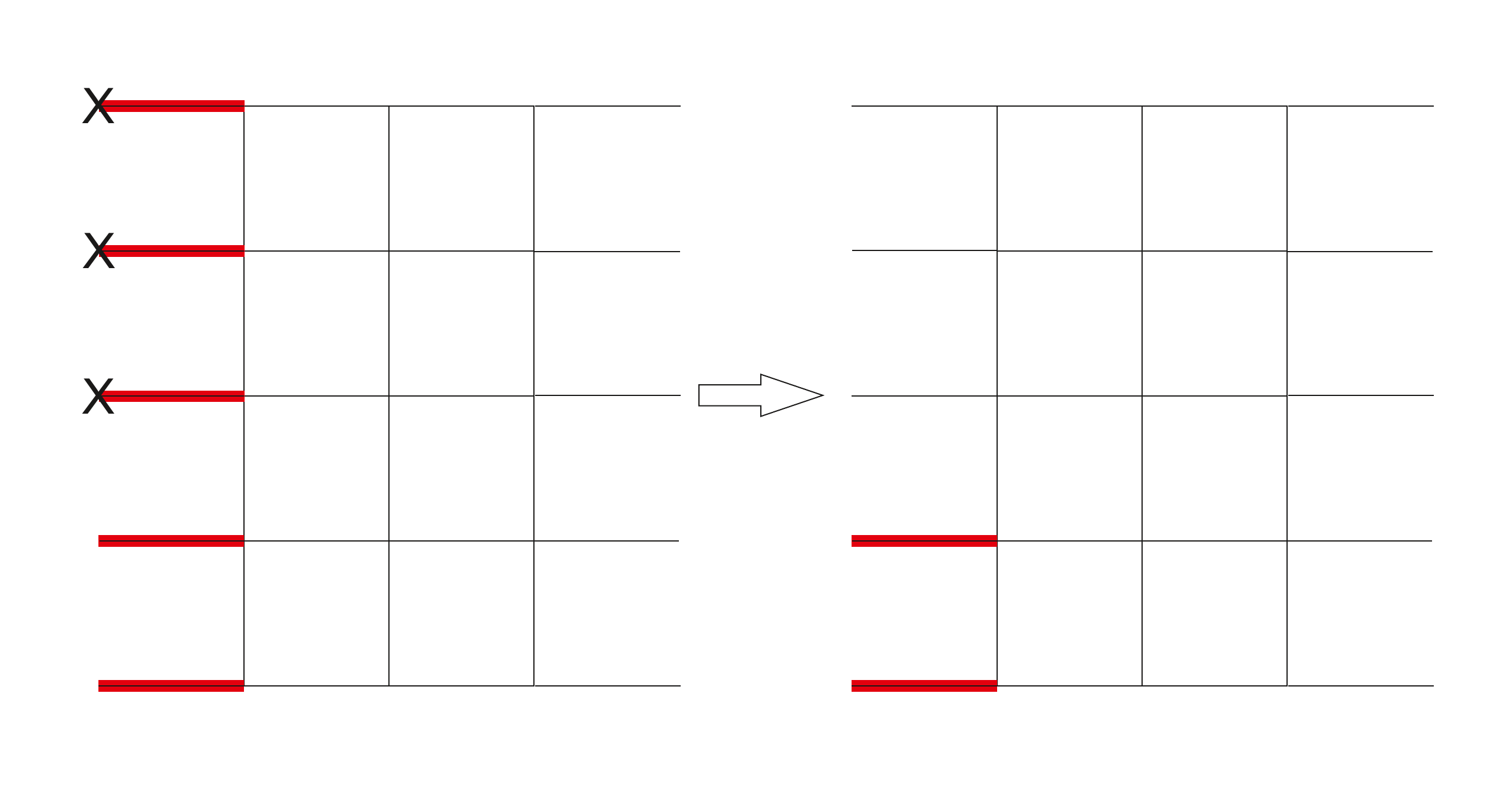}
\caption{\label{figura10b}
 Removal of a path connecting two opposite boundaries without periodicity conditions by switches.}
\end{figure}

Pairs of defects are also typically connected by similar chains of $(01)$ permutations. 
We can think of those chains of permutations as paths and cycles in the dual lattice. They can be seen as an analogue of paths describing logical operators in a Kitaev code on a torus  (see Fig. \ref{figura10b}). However, in the case of the planar grid there is a notable difference between the game and the code defined by the same labeled graph. In Kitaev codes it is not possible to remove a path of errors on the sharp boundary of the code (because the removal here can be performed by application of stabilizer operators, and there are no stabilizer operators that act on a single qubit as in this setting this would violate the demand for the stabilizers to commute with each other). On the other hand, on the graph associated with a $d=2$ non-local game it is possible to remove such a path. The possibility of destroying a path connecting two boundaries without periodicity condition is already a result of the existence of a spanning tree that joins all the boundaries without periodic conditions (cf. Fig. \ref{figura10}).

Let us focus on equivalence conditions for two non-planar games. Each of the games can be characterized by a set of (01) permutations labeling the edges. These correspond to a set of paths in the dual lattice and we will call them \textit{error paths}, for the analogy with quantum error correction codes.
We will prove the following:

\begin{theorem}\label{Th8}
Two labelings $K_1, K_2$ of a graph are equivalent iff the labeling $K$ defined as $K(e)=K_1(e)K_2(e)^{-1}$ for all edges is equivalent to $K_{\id}$ defined as $K_{\id}(e)=\id$ for all $e$.
\end{theorem}

\begin{proof}

Let $K(e)=K_1(e)K_2(e)^{-1}$ be equivalent to $K_{\id}.$ By the definition of equivalent labelings, one can transform $K$ into $K_{\id}=KK^{-1}$ through switches. If the same switches are applied to $K_1$, it is transformed into $K_1K^{-1}=K_2$. 

Now assume that the labelings are equivalent. Then $K_2$ can be obtained from $K_1$ through switches. Applying the same set of switches to the labeling $K$ transforms it into $K_{\id}.$

\end{proof}

An easy way to check that the labeling $K$ is equivalent to $K_{\id}$ is the following.

\begin{obs}
Let $P$ be the set of error paths of a labeling $K$ with $S_2$. The labeling is equivalent to $K_{\id}$ iff the number of loops in each homology class is even and all other error paths are contractible to a point.
\end{obs}

\begin{proof}
A pair of loops within the same homology class can be annihilated through switches by moving one onto the other. Obviously, an error path is considered contractible if and only if it can be removed by switches.

On the other hand, an unpaired loop cannot be removed by switches and thus a labeling with an odd number of loops in some homology class is not equivalent to $K_{\id}.$
\end{proof}

From the above we see that a necessary condition for two games two be equivalent is that they have the same set of defects. Otherwise, error paths in $P$ will be unconctractible to a point. The remaining mechanism to generate paths in $P$ that are uncontractible is solely associated with the topology of the surface on which the graph is defined. Namely, a lattice (e.g. on a cylinder or a torus) can allow for construction of paths that cannot be contracted to a point by a continuous transformation  -- a set of paths with this property, that can be transformed into each other, is called a \textit{homology class}. Therefore, games with error paths belonging to different homology classes are not equivalent. From the definition of equivalence of labelings we see that, as local transformations do not change homology class of the paths, nor the position of defects, they have to be the same for games belonging to the same equivalence class. Therefore

\begin{prop}
Two labelings of a $K_1$, $K_2$ of a grid are equivalent iff they have the same set of defects and the labeling $K=K_1K_2^{-1}$ contains no loops which cannot be annihilated by switches.
\end{prop}

The role of the topological properties of the surface on which the graph is defined in the equivalence between two games  is based on the fact that switches which transform equivalent graphs into each other can deform error paths only in a continuous manner. Therefore, two equivalent games cannot have different sets of loops within each of the homology classes admitted by the geometry. This property does not depend on $d$, and therefore will be crucial for classification of games for higher number outcomes.

\section{Classification of non-locality games for $d=3$}\label{3}
For games with $d=3$ possible measurement outcomes (and the same number of different correlations that can be demanded for outcomes of a pair of measurements), the group $S_{3}$ can be generated by $L_{3}$, where $L_d=\{\tilde{\pi}_i:\tilde{\pi}_i(x)=i-x \mod d\}$. A composition of even number of permutations from $L_d$ forms a permutation belonging to a subset $L_d'=\{\tilde{\sigma}_i:\tilde{\sigma}_i(x)=i+x \mod d\}$ of $S_{d}$. Also, $S_3=L_3\cup L_3',$. Note that for $d=2$, this structure degenerates, as $L_{2}=L_{2}'$. 

Before going to classification of these games, we prove some useful statements about graphs labeled by permutations from $L_{d}$ and $L_{d}'$.

\begin{obs}
\label{Oequal}
For a bipartite graph labeled with permutations from $L_{d}$ ($L_{d}'$), there exists an equivalent labeling with permutations from $L_{d}'$ ($L_{d}$).
\end{obs}

\begin{proof}
Because the graph is bipartite, we can divide its vertices into two disjoint sets. Applying switches with permutations from $L_{d}$ to vertices from one of these sets will transform permutations from $L_{d}$  ($L_{d}'$) that label edges into permutations from $L_{d}'$ ($L_{d}$). 
\end{proof}

For games that can be described by labeled graphs on square lattices, we take into account situations where all edges are labeled with permutations from $L_{d}$, or equivalently from $L_{d}'$. This implies that the defect class of each cell is a permutation from $L_{d}'$. Nevertheless, many of the results can be easily generalized to different types of lattices or to all planar graphs. Furthermore, Observation \ref{Oequal} implies that the results can also be (indirectly) applied to graphs labeled with $L_d.$

In the case of $d>2$, because the size of the group $L_{d}'$ is larger than $L_{2}'$, each cell can have one of $d$ different defect classes. It follows from the proof of Proposition \ref{Tdodawanie_n2} that two graphs labeled by $L_{d}'$ have the same sets of defects of each class then they have the same sets of cycles of each class. 

The lemmas \ref{L1} and \ref{L2} from the Appendix applied to $d=3$ are useful to derive the theorem for equivalence of games with correlations defined by permutations from $L_{3}$ or $L_{3}'$.

\begin{theorem}
\label{Tn3}
Two labelings with $L_{3}'$ are equivalent iff either the defect class of each cell is the same for both labelngs or 
the set of defects of class $x$ for one labeling is the set of defects of class $-x$ for the other labeling, for every $x\in\{0,1,2\}$.
\end{theorem}

\begin{proof}
($\Rightarrow$ part).
Two labelings are equivalent if we can obtain one from the other through switches. Select a cell of defect class $x$ and let $v_0$ denote the starting vertex of the characteristic permutation. We apply a switch $s(v,\pi)$:

\begin{enumerate}
\item If $v\neq v_0$, then the defect class of the cell remains unchanged, due to Lemma \ref{L1} and Lemma \ref{L2}.
\item If $v = v_0$ and $\pi=\tilde{\sigma}_{i}\in L_3',$ then the class of the cell remains unchanged, due to Lemma \ref{L1}.
\item If $v = v_0$ and $\pi=\tilde{\pi}_{i}\in L_3,$ then the defect class of the cell after the switch is changed into $-x$, due to Lemma \ref{L2}.
\end{enumerate}

Since $S_3=L_3\cup L_3',$ a switch can only change the defect class of a cell from $x$ to $-x.$

If a switch by $\pi\in L_3$ is applied to some vertex $v$ of $G$, we obtain a labeling in which some edges are labeled with $L_3$. In order to return to $L_3'$ we must switch all vertices adjacent to $v$ by some permutation $\pi\in L_3.$ This shows that if we switch a vertex by some $\pi\in L_n,$ we must also apply such switches to all vertices of the graph. Hence, the defect class of every cell is changed from $x$ to $-x.$ Otherwise we only switch by permutations $\sigma\in L_3'$ and and the classes of all cells remain unchanged.
\end{proof}

\begin{proof}
($\Leftarrow$ part) Same as in the proof of \ref{T:any_n}, which can be found in the Appendix. 
\end{proof}

\section{Classification of non-locality games for $d\geq 4$}
\label{sub:4+}

In Theorem \ref{T:any_n} we present a generalization of Theorem \ref{Tn3} to a graph with an arbitrary number of outcomes and labeled with $L_{d}'$. 

\begin{theorem}
\label{T:any_n}
Two labelings $K,L$ of a connected planar graph with permutations from $L_{d}'$ are equivalent iff there exists a permutation $\pi\in S_d$ such that for every cell $c$ of a graph we have $cl(c,L)=\pi cl(c,K)\pi^{-1}$, where $cl(c,L)$ is a defect class of a cell $c$ in labeling $L$.
\end{theorem}

The proof of this theorem can be found in the Appendix.

From the above we see that any transformation which connects two equivalent labelings with $L_{d}'$ has to be composed from switches $s(v,\sigma_{1}(v)\pi\sigma_{2}(v))$, such that $\sigma_{1}, \sigma_{2}\in L_{d}'$ and may depend on the vertex $v$, while $\pi$ is the same for all vertices and in general does not have to belong to $L_{d}'$. 

In Section \ref{periodic2} we provided conditions for equivalence of labelings with $L_{2}'$ on different surfaces. Here we generalise these results to arbitrary $d$. A switch $s(v,\tilde{\sigma}_{i}\in L_{d}')$ on an arbitrary vertex of the network continously deforms the labelings, in a sense that, in a convention in which permutations on the adjacent edges point from the vertex the the exterior, every permutation $\tilde{\sigma}_{x}$ on an edge is shifted: $\tilde{\sigma}_{x}\rightarrow\tilde{\sigma}_{i}\tilde{\sigma}_{x}=\tilde{\sigma}_{i+x}$. Let us define an operation of adding labelings $L$, $K$ from $L_{d}'$ of a given graph $G=(V,E)$. The sum $M=K+L$ is a labeling in which $M(e)=K(e)L(e)$ for every $e\in E.$ The difference $K-L$ will denote the sum $K+L^{-1}.$

\begin{lemma}
\label{addswitch}
Let $K_s:E\mapsto L_d'$ be a labeling obtained from $K_{\id}$ through a single switch $s(v,\pi).$ Then the labeling $M=K+K_s$ is equivalent to $K:E\mapsto L_d'$ and adding $K_s$ to a labeling has the same effect as performing the switch $s(v,\pi).$
\end{lemma}

\begin{proof}
The labeling $K_s$ assigns $\id$ to all edges not incident with the vertex $v$, $\pi$ to all edges $vu$ coming out of $v$ an $\pi^{-1}$ to all edges $uv$ going into $v$. It follows that 

\be
M(e)= \left\{\begin{array}{lcr}
K(e)\pi & if & e = vu\\
K(e)\pi^{-1} & if & e = uv\\
K(e) & & otherwise.
\end{array}\right.
\ee

Since the permutations in $L_d'$ commute, $K(e)\pi^{-1} = \pi^{-1}K(e)$ and thus $M$ is the labeling obtained from $K$ through the switch $s(v,\pi).$
\end{proof}

From this, we have the following.

\begin{cor}
\label{addswitches}
Let $K_1,K_2:E\mapsto L_n'$ be two labelings of a graph. $K_2$ can be obtained from $K_{\id}$ through switching by permutations from $L_d'$ iff $K_1+K_2$ is equivalent to $K_1.$
\end{cor}

Therefore, we arrive with the following equivalence criterion for all planar games with permutations from $L_{d}'$ for arbitrary $d$:

\begin{theorem}
\label{Th15}
Let $K_1$ and $K_2$ be two labelings of a graph with $L_d'$. The labelings are equivalent iff there exists a labeling $M$ and permutation $\pi$ such that $K_1=\pi(K_2)+M,$ where $M$ is equivalent to $K_{\id}$ and  and $\pi(K)$ denotes the labeling obtained from $K$ by applying the switches $s(v,\pi)$ to all vertices $v\in V.$
\end{theorem}

\begin{proof}
$(\Rightarrow)$ 
Let $K_1$ and $K_2$ be two equivalent labelings. It follows from Corollary \ref{cor:any_n} that there exists a permutation $\pi$ such that $K_1$ and $\pi(K_2)$ have the same set of defects of each class. We can now transform $\pi(K_2)$ into $K_1$ using only switches from $L_d'.$ Thus, by Corollary \ref{addswitches}, we have $K_1=\pi(K_2)+M$ for some labeling $M$ equivalent to $K_{\id}.$

$(\Leftarrow)$
Now assume that $K_1=\pi(K_2)+M,$ where $M$ is a labeling obtained from $K_{\id}$ through a switches with permutations from $L_d'.$ This means that $\pi(K_2)$ can be transformed into $K_1$ through the same set of switches. Thus, $\pi(K_2)$ is equivalent to $K_1.$ Since $\pi(K_2)$ is equivalent to $K_2,$ it means that $K_1$ and $K_2$ are equivalent.
\end{proof}

From the above we can obtain yet another equivalence condition, analogous to the one given in Theorem \ref{Th8} for labelings with $d=2$.

\begin{cor}
\label{minus}
Two labelings $K_1, K_2:E(G)\mapsto L_d'$ are equivalent iff, there exists a permutation $\pi\in S_d$ such that the labeling $K=K_1-\pi(K_2)$ is equivalent to $K_{\id}.$ 
\end{cor}

In order to characterize equivalence classes of labelings with $L_{d}'$ permutations on a surface of genus $2$ or more, we make use of graph topological properties. First we define the \textit{enlarged spanning tree} $T^{L}$ as follows:

\begin{dfn}
For a graph $G$, let us enlarge a set of edges $T^{0}$ forming a spanning tree $T(G)$ by adding an edge which belongs to a cell such that this edge is the only member of the cell not belonging to $T$. Update the set $T^{0}$ with this enlarged set. The above procedure is to be applied unless there is no cell that contains only one edge not belonging to $T^{0}$. Such a set: $T^{L}=T^{0}$ will be called the enlarged spanning tree of $T$.    
\end{dfn}  

Note that if the graph $G$ is planar, then $T^{L}=G$. On a torus/surface with a higher genus some edges of the graph remain outside of $T^{L}$.  

\begin{lemma}
\label{L:tree+}
The complement $G-T^{L}$ of the enlarged spanning tree $T^{L}$ is composed of all edges of a graph that do not belong to $T$ and belong to some nontrivial loops.  
$G-T^L$ contains exactly one nontrivial loop in each homology class.
\end{lemma}

As an analogue to a notion of the class of a cell, we define the class of a loop defined on the complement of $T^{L}$ to be equal to the index of a permutation from the set $\tilde{\sigma}_{0},\tilde{\sigma}_{1},\tilde{\sigma}_{2},\dots,\tilde{\sigma}_{d-1}$ that labels the edges belonging to the loop.

\begin{theorem}\label{Prop17}

Equivalence class of a game of permutations from $L_{d}'$, defined on a graph with a spanning tree $T$ on a surface with genus $g$, is completely characterized by: 
\begin{itemize}
\item sets of cells of classes $cl=1,2,\dots,d-1$,
\item sets of loops on the complement of $T^{L}$ of classes $cl=1,2,\dots,d-1$, \\ \\
where all the above sets are described modulo a chosen permutation $\pi\in S_{d}$ and $\pi\notin L_{d}'$ that maps permutations from $L_{d}'$ into $L_{d}'$.  Such permutation maps defect classes and permutations on loops belonging to the complement of $T^{L}$ in the same way: $cl\rightarrow \pi cl \pi^{-1}$, $\tilde{\sigma}_{i}\rightarrow \pi\tilde{\sigma}_{i}\pi^{-1}$. Number of the non-equivalent loops of permutations is equal to $d^{2g}$.

\end{itemize}

\end{theorem}

In the proof of the above Theorem we will use the fact that canonical form of a $L_{d}'$ game associated with a given spanning tree $T$ depends on permutations that belong to $T^{L}$, but not to $T$ (they are fixed, up to permutation $\pi$ applied to every vertex, by classes of defects), as well as by permutations that belong to the complement of $T^{L}$ (they are partially fixed, up to permutation $\pi$ applied to every vertex, by classes of defects, but also provide some additional characteristic of a game arising from periodic boundary conditions).

\section{Classical values from defects and loops}\label{6}

In this section we discuss some methods for calculating the classical value of a game based on the properties of the corresponding labeled graph. On planar graphs, the classical value is given by error correction algorithm in Kitaev code. A modification of the algorithm is necessary when periodic boundary conditions are enforced. Since the sets of defects and loops of each class are enough to define a game, one could use these properties to calculate its classical and quantum values. 

The values are winning probabilities given by the formula
\begin{equation}\label{win}
p_{win}=\max_{P}\sum_{x\in A,y\in B}p(x,y)\sum_{a,b\in [0,\dots,d-1]}V(ab|xy)p(ab|xy).
\end{equation}
Above, $V(ab|xy)$ is the winning condition, taking value $1$ iff $\pi_{xy}(a)=b$, where $\pi_{xy}$ is a permutation between the responses for a pair of questions, and fixed by the unique game, $p(ab|xy)$ is the conditional probability of obtaining a pair of answers $(a,b)$ given questions $(x,y)$, and $p(x,y)$ is the distribution of questions. Below, we take it uniform. Maximization is performed over families of conditional probabilities. In classical case, the optimization can be performed over \textit{deteministic} local hidden variable models, i.e. where we have $\sum_{b}p(ab|xy)=p(a|x)$ for every $x$ and $y$, and that $p(a'|x)=1$ for a selected response $a'$ \cite{Fine1982,Brunner2014} (the same applies to Bob responses). Therefore for the games investigated in this paper, classical values can be calculated through search over all possible classical assignments of values $[0,\dots,d-1]$ to nodes of the corresponding graph. In a case of optimal labeling, the classical value is going to be $p_{win,cl}=\frac{|E(G)|-\beta_{C}}{|E(G)|}$, where $|E(G)|$ is the number of pairs of questions (number of edges of the corresponding graph), and $\beta_{C}$ is the number of contradictions, i.e. number of winning conditions $\pi_{xy}$ that are not satisfied by the optimal labeling.   

Here we show that in the case of $d=2$ outputs the algorithm for calculating the classical value can be efficient. The method is similar to ones used in error correction and relies on connecting pairs of defects with the shortest possible paths in the dual lattice to minimize the number of contradictions. We also consider the case with $3$ or more outputs and propose an analogous algorithm. In this case defects may have to be gathered into sets of more than two and connected with minimum Steiner trees in the dual lattice instead of shortest paths. A \textit{Steiner tree} for a given set $S$, as defined in \cite{Hwang1992}, is a tree which connects all vertices in $S$. Figure \ref{defectsets} provides examples of such paths and trees.

First notice that the labeling $K_{\id}:E\mapsto L_d'$ which assigns $\id$ to all edges of the graph gives rise to no contradictions. In fact, every labeling with no contradiction is equivalent to $K_{\id}$. We can use this fact to prove the following lemma.

\begin{lemma}
For every labeling $K$ of a graph $G$ there exists an equivalent labeling $K_{o}$ such that $K_o(e)\neq\id$ for exactly $\beta_C(G,K)$ edges. No labeling equivalent to $K$ assigns $\id$ to more than $\left|E(G)\right|-\beta_C(G,K)$ edges.
\end{lemma}

\begin{proof}
Notice that a labeling $K$ has no contradictions if and only if it is equivalent to $K_{\id}$. 
It follows that $\beta_C(G,K)=k$ if and only if changing the labels of a certain set $S$ of $k$ edges results in a labeling $K'$ equivalent to $K_{\id}$. The labeling $K_{\id}$ can be obtained from $K'$ through a specific set of switches. Applying the same set of switches to $K$ results in a labeling $K_o$ which assigns $\id$ to an edge $e$ if and only if $e\notin S.$

Obviously any labeling in which identity is assigned to more than $\left|E(G)\right|-\beta_C(G,K)$ edges has fewer than $\beta_C(G,K)$ contradictions and thus it is not equivalent to $K$. 
 
\end{proof}

We will call $K_o$ an \textit{optimal labeling} and many of our methods will involve finding an optimal labeling for a given configuration of defects and loops.

The simplest case is a torus (or higher genus surface) with no defects. In this case for each loop of length $k$ the graph contains a set of $k$ disjoint bad cycles, each containing one edge of the loop. Even if two loops intersect, the cycles associated with those loops are still edge-disjoint. Thus, the contradiction number is the total length of all loops. This does not depend on the number of outputs or the classes of the loops.

\subsection{Perfect matching approach for $d=2$}
Now we consider a game with $d=2$ outputs and no loops. In this type of game optimal labelings are defined by sets of paths in the dual lattice $G^*$ such that a) each path connects two defects and b) the total number of edges in the set of paths is as small as possible. Such a set is equivalent to a perfect matching with maximum weight in a weighted graph $G'$ defined as follows.

\begin{enumerate}
\item $G'$ is a complete graph and $V(G')$ is the set of defects of $(G,K).$
\item The weight of an edge $uv$ is $w(uv)=\diam(G^*)-d(u,v)$, where $d(u,v)$ is the length of a shortest path in the dual lattice connecting the defects $u$ and $v$.
\end{enumerate}

In short, the method consists of the following steps:
\begin{enumerate}
\item Find the distance (i.e. the length of the shortest path) between each pair of defects;
\item Choose the pairs so that: a) every defect has a pair and b) the total length of the paths is minimized. 
\end{enumerate}

Finding a maximum weight perfect matching is a well studied problem and a polynomial time algorithm for solving it can be found in \cite{Ed65}. It is a generalization of a method described in \cite{Kuhn55} and \cite{Mun57}. The method relies on Berge's lemma \cite{Berge57}, which allows us to increase the size of a non-maximum matching using augmenting paths, as well as a method known as blossom shrinking. A blossom is defined as an odd cycle in $G'$ with a maximum matching. Shrinking all blossoms in the graph, i.e. replacing them with single vertices whose neighborhood is the same as the neighborhood of the blossom, allows us to apply the algorithm for bipartite graphs to all graphs. 

\begin{figure}[h!]
\includegraphics[scale=0.2]{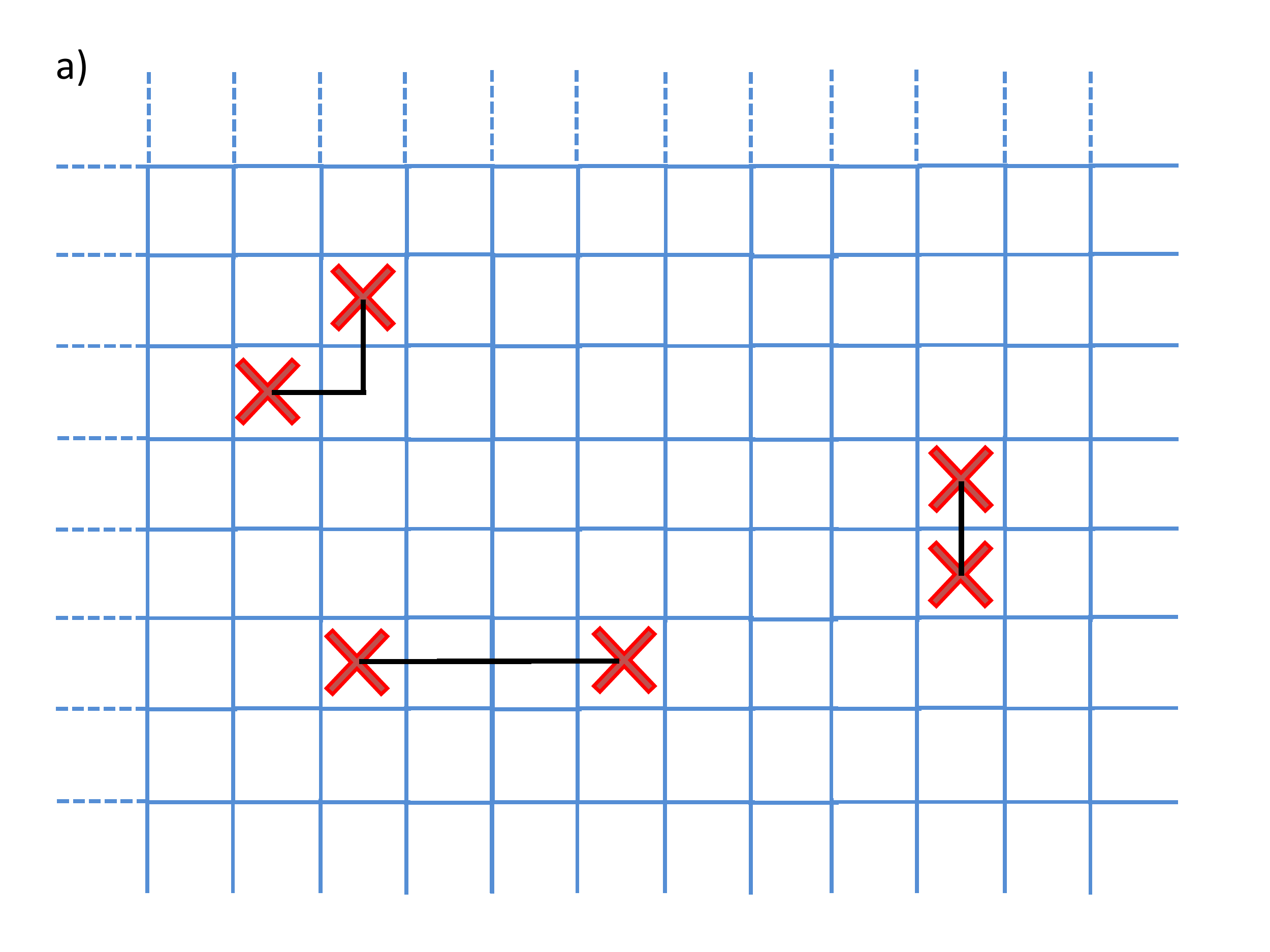}
\includegraphics[scale=0.2]{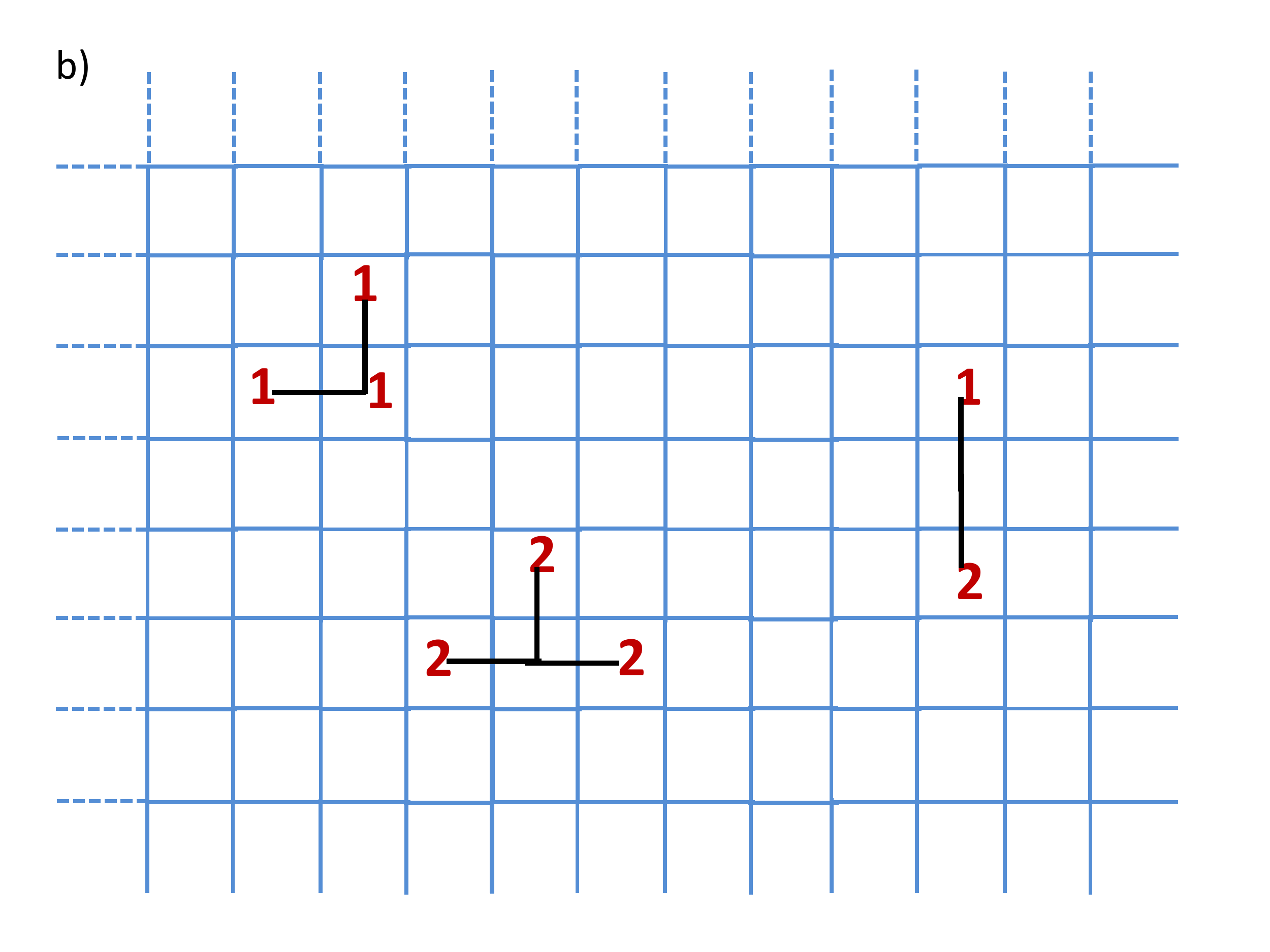}
\caption{\label{defectsets}
a) Pairs of defects connected by paths for $d=2$, b) Sets of defects connected by trees for $d=3$.\\Notice that for $d=2$ all defects are of class $1$ and thus a pair of defects always adds up to $0 \mod 2$. }
\end{figure}

In the case of labelings on a non-planar grid, where loops are a possibility, some additional steps are needed to determine the classical value. From the above method we obtain a minimal labeling $M$ with the smallest set of non-identity edges possible for a given configuration of defects. However, it is still possible that $M$ is not equivalent to the original labeling $K$. In this case, we have $K-M = L$, where $L$ is a labeling with no defects, but containing one or more non-trivial loops. Thus, to find an optimal labeling equivalent to $K$, we must add the necessary loops to $M.$ In order to ensure that the resulting labeling is indeed optimal, we must add a set $S$ of loops in such a way that:

\begin{enumerate}
\item $S\equiv L$,
\item $\left|S-M\right|+\left|M-S\right|$ is minimized.
\end{enumerate}

This way we obtain a labeling $K_o$, which is equivalent to $K$. 
We now show that this labeling is indeed optimal. 

Suppose some labeling $K_o'$ is equivalent to $K$ and has fewer non-identity edges than $K_o.$ Such a labeling is obviously also equivalent to $K_o$ and, as such, differs from $M$ by the same set of loops. Thus, $K_o'$ is a labeling obtained from $M$ by adding a set $S'\equiv S$ to the minimal labeling $M$. But if $S$ minimizes $\left|S-M\right|+\left|M-S\right|$, then $K_o'$ cannot have fewer non-identities than $K_o.$

\subsection{Methods for $d\geq 3$}
A generalization of the above method can most likely be used for $d\geq 3.$ In this case the defects are not paired, as they were for $d=2$, but grouped into sets of up to $d$ elements such that the classes of defects in each set add up to $0 \mod d$.
Here the steps are:
\begin{enumerate}
\item Find all minimal sets of defects which add up to $0$ mod $d$;
\item For every such set, calculate the minimum number of edges in a tree connecting them in the dual lattice;
\item Choose the sets so that: a) every defect is in exactly one set and b) the total length of the trees is minimized. 
\end{enumerate}

Note that the above generalization resembles generalization of decoding algorithm for qubit surface codes into qudit systems \cite{Watson2015}. However, first step of our algorithm is applied globally to the whole structure of the graph, in constrast with grouping defects into local clusters of increasing size, which in \cite{Watson2015} is performed in spirit of renormalization group approach \cite{Bravyi2011}.
 The final step of the algorithm is equivalent to finding a maximum weight perfect matching in the weighted hypergraph $G'$ defined as follows:

\begin{enumerate}
\item $V(G')$ is the set of all defects in $(G,K).$
\item A set $S\subset V(G')$ is an edge iff it is a minimal set such that $\sum_{x\in S}cl(x)=0$ (mod $d$).
\item The weight of an edge is $w(S)=T(G^*)-St(S)$, where $St(S)$ is the minimum length of a Steiner tree for the set $S$ and $T(G^*)$ is the number of edges in the spanning tree of the dual lattice. 
\end{enumerate}

By a perfect matching in a hypergraph we mean a set $M$ of edges such that each vertex belongs to exactly one edge in $M$. One problem with such a generalization is properly defining the hypergraph equivalent of a blossom. The problem of Steiner trees is also NP-hard in general,  which may increase the complexity of the algorithm for large numbers of outputs. Nevertheless, an interesting implication is the fact that the hypergraph $G'$ is guaranteed to have a perfect matching.

\subsection{Spanning tree approach}
Another approach, which may be better for games with a large number of outputs, relies on the following result.

\begin{lemma}
Every optimal labeling $K_o$ of a grid $G$ is a canonical representation for some spanning tree of $G$.
\end{lemma}

\begin{proof}
Let $K_o$ be an optimal labeling of a grid $G$. First, remove any edges from $G$ which define a loop in $K_o.$ Next, remove an arbitrary set of edges which defines a loop in each direction in which there is no loop in $K_o$. The remaining graph $H$ is a planar grid. 

The set $S$ of edges $e\in E(H)$ such that $K_o(e)\neq \id$ corresponds to a certain set $S^*$ of edges in the dual lattice $H^*$ of $H$. The set $S^*$ is a subset of the edge set of a spanning tree $T$ of $H^*$. If we remove all edges in $H$ corresponding to edges of $T$, what remains is a spanning tree $T_1$ of $H$ such that $K_o$ assigns $\id$ to all of its edges. Since any spanning tree of $H$ is also a spanning tree of $G$, It follows that $K_o$ is the canonical representation with respect to the spanning tree $T_1$.
\end{proof}

 It is easy to see that a game has only one canonical representation for any given spanning tree. This representation is easy to obtain from the tree itself using the tree enlarging procedure from section \ref{sub:4+}. Unfortunately, not every canonical representation is an optimal labeling for a given set of defects and loops. Therefore to find the optimal labeling we need to compare canonical representations with respect to different spanning trees. The number of spanning trees grows exponentially with respect to the number of vertices in the grid, so this method appears to be NP-hard. However, unlike the algorithm based on matchings, its complexity does not depend on the number of outputs.

\section{Examples}\label{ex}

Below we show how local and non-local conditions, associated with topology of a given game, can influence maximal probabilities of winning the games, with classical and quantum resources provided. Classical values can be calculated exactly using methods described in the previous section.
In the quantum case, the optimization in (\ref{win}) is performed over all states $|\Psi\rangle$ in a Hilbert space and sets of ortogonal projective measurements on Alice and Bob subsystems such that $p(ab|xy)=\langle\Psi |M_{x}^{a}M_{y}^{B}|\Psi\rangle$, and $[M_{x}^{a},M_{y}^{b}]=0$. Due to the fact that dimension of Hilbert space is not fixed, every POVM measurement can be described in the above form. The task of finding the optimal protocol saturating the quantum value is equivalent to finding a Bell operator $S=\sum_{a,b,x,y}V(ab|xy)M_{x}^{a}M_{y}^{b}$ with maximal norm. 

Upper bound on a quantum value can be found by solving a semidefinite problem resulting from relaxing the conditions for set of quantum probabilities aimed at satisfying restrictions of a particular game. An infinite series of these relaxations, in a form of the so called NPA hierarchy \cite{Nav2008} of semidefinite programs, was shown to give values converging to quantum value as defined above. Here, we will calculate the upper bound $Q^{\uparrow}$ on quantum values by solving first level problems of the hierarchy, i.e. by optimizing (\ref{win}) under the restriction that the matrix $\Gamma$ with entries $\Gamma_{ij}=\langle|\Psi O_{i}^{\dagger}O_{j}|\Psi\rangle$ is semidefinite positive, where the set $O=Id\cup\{M_{x}^{a}\}_{a,x}\cup\{M_{y}^{b}\}_{b,y}$ is composed of single orthogonal projective operators (and Identity).

A lower bound $Q^{\downarrow}$ is calculated on the basis of the so-called see-saw algorithm \cite{Werner2001}. The algorithm, for a given dimension of Hilbert space, is based on a sequence of semidefinite programs that aim at looking for an optimal Bell operator $S$ based on linear structure of $S$ with respect to measurement of Alice when measurements of Bobs are fixed (and vice-versa). It starts from randomly chosen measurement operators for Alice and Bob and an initial state shared between them, and then performs 3 steps: for given measurement operators, finds a state that leads to maximal winning probability; ii) optimizes measurement operators for Alice given quantum state and measurement operators for Bob; iii) optimizes measurement operators for Bob given quantum state and measurements operators for Alice. At each step, an SDP program is invoked to solve the optimization problem. For a given initial random configuration, this 3 step process is repeated 20 times, and a lower bound on quantum value presented here is calculated as maximum over 20 independent runs of the program. 

While it is not even known if, for a fixed dimension of the Hilbert space, the see-saw algorithm allows for convergence into the global maximum, for a qubit case we obtain exact quantum values of the games, as both upper and lower bounds coincide. We found that the analytical upper bound on quantum value of a d-outcome XOR game \cite{Ramanathan2016} coincides with the one calculated by an SDP in the first level of the NPA hierarchy (for games with periodic boundary conditions), or exceeds it (for games without periodic boundary conditions). Therefore, its values are not reported. 

We start with the problem of boundary conditions, for a XOR game (n=2).  In Fig. \ref{fig:grid} the grids a) and b) without boundary conditions are equivalent. Since a) is equivalent to a grid with identity assigned to all edges, it contains no contradictions. Thus, both classical and quantum winning probability for this graph is equal to $1$.

\begin{figure}[h!]
\includegraphics[scale=0.30]{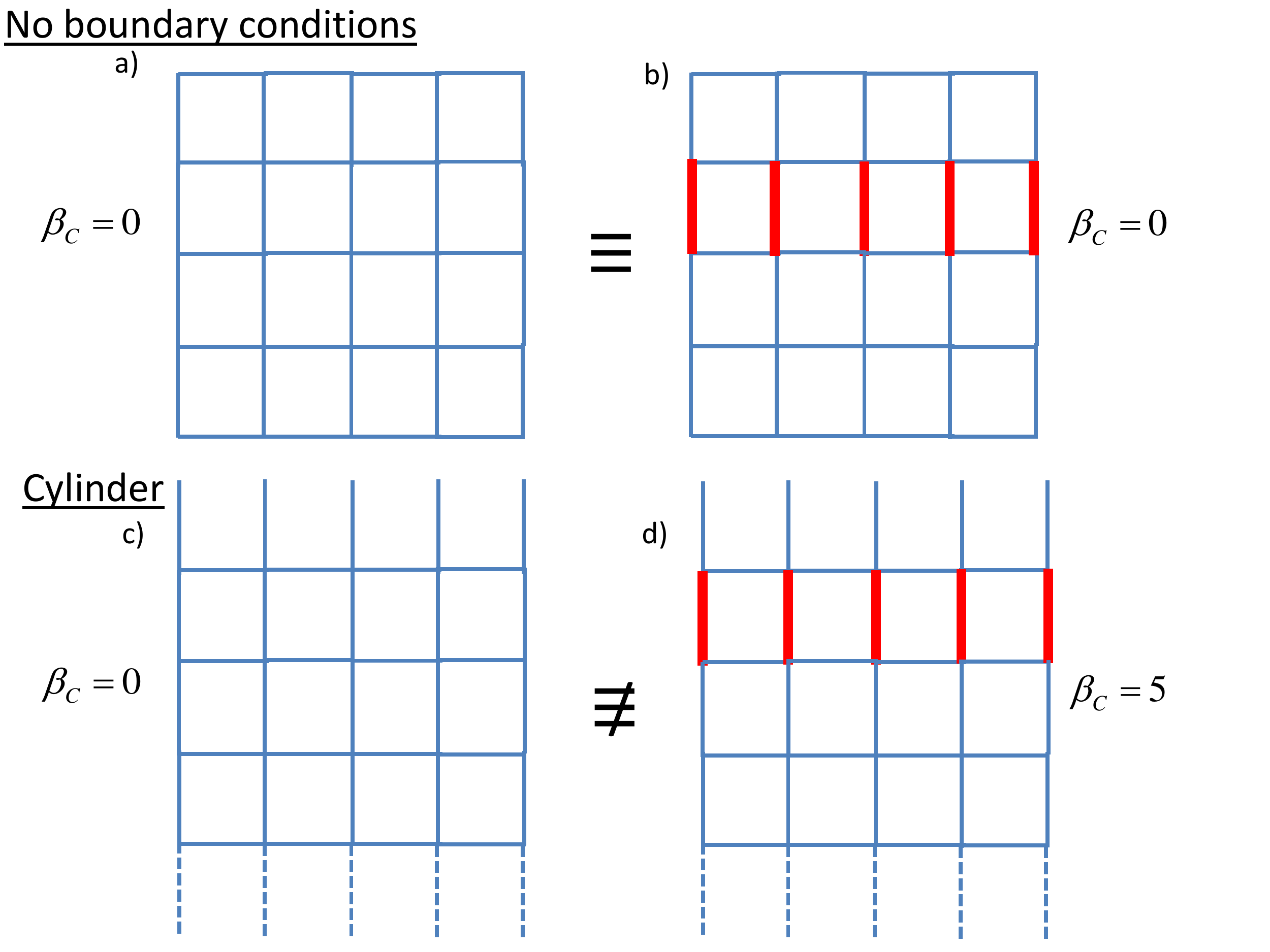}
\caption{A $4\times 4$ grid with and without periodic boundary conditions. Blue edges are associated with $\id$ permutations, while red edges are associated with a permutation $(01)$.  No defects present. k+1 homologically non-trivial bad cycles cross red edges (d).}
\label{fig:grid}
\end{figure}

On the torus a loop can occur which cannot be removed by switches. Hence the graphs in Fig. \ref{fig:grid} c) and d) are not equivalent. In the case of a $k\times k$ torus with one loop and no defects $\beta_C=k$, regardless of the number of outputs or the class of the loop. This is because the graph contains a set of $k$ disjoint chordless bad cycles, one for each edge of the loop. These cycles behave like defects and cannot be removed by switches. Since the cycles are disjoint, one edge needs to be changed in order to make each cycle good. Hence, the classical winning probability for this game is $p=\frac{\left|E\right|-\beta_C}{\left|E\right|}=\frac{2k^2-k}{2k^2}=\frac{2k-1}{2k}$.

\begin{figure}[h!]
\includegraphics[scale=0.4]{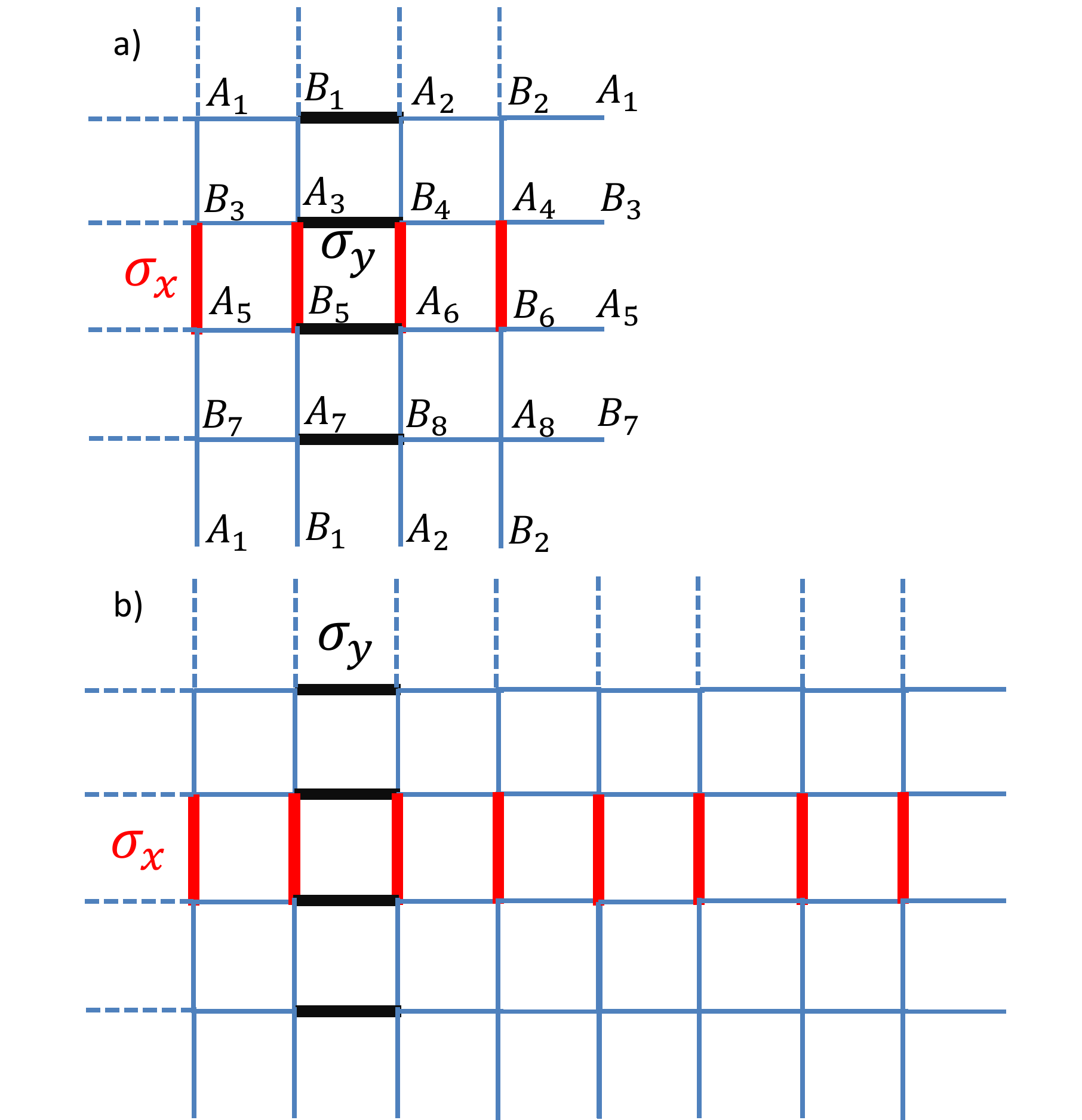}
\caption{Two non-locality games of different sizes with periodic boundary conditions, with permutations $\sigma_{x}$ on red edges, $\sigma_{y}$ on black edges, and Identities on blue edges. Classical and quantum values given in Table \ref{Tabel1}. Nodes corresponding to different questions $A_{i}$ and $B_{j}$, $i,j=1,\dots,8$ are marked on the smaller graph.}
\label{fig:lines}
\end{figure}

\begin{figure}[h!]
\includegraphics[scale=0.30]{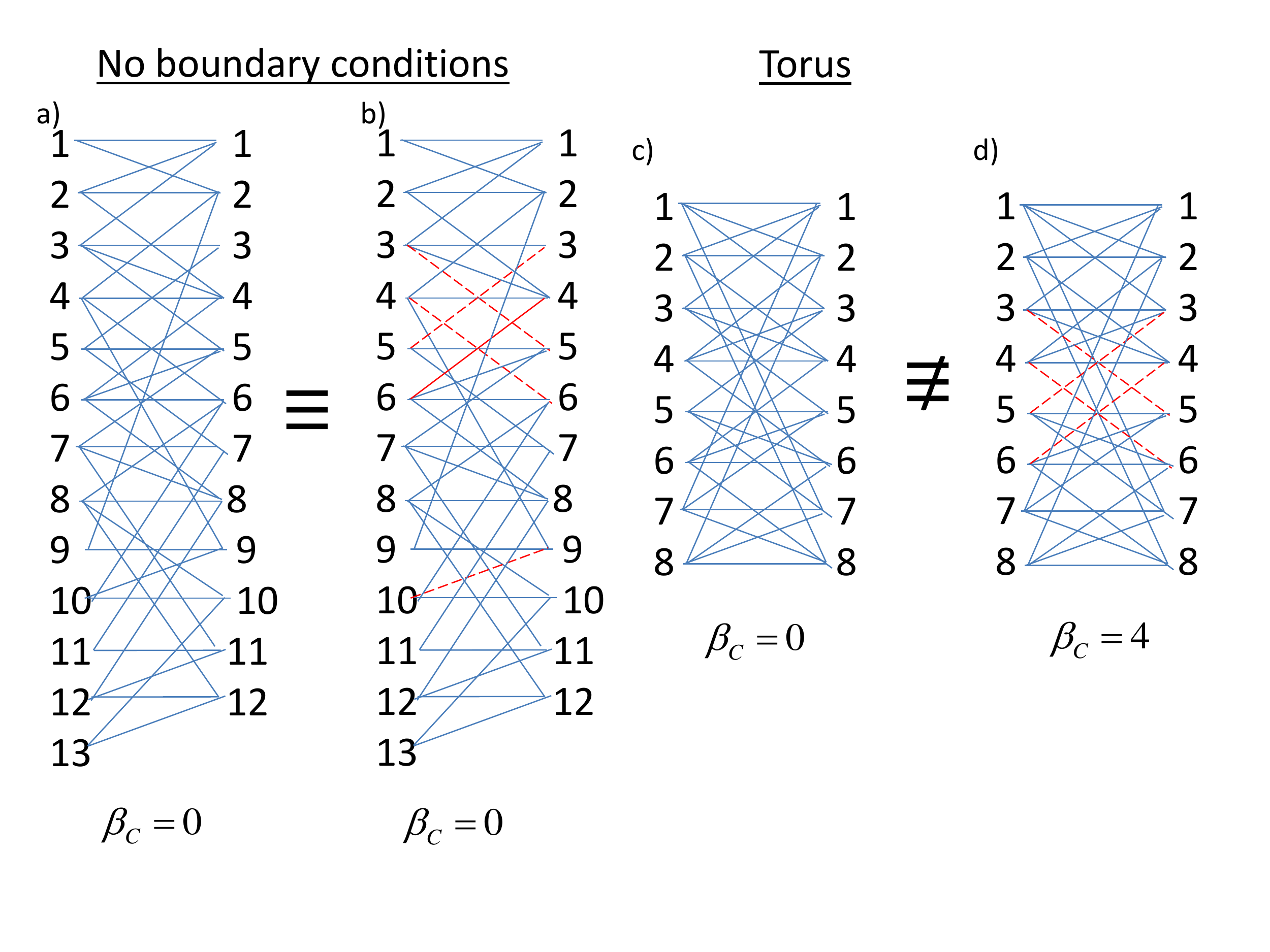}
\caption{The grids from Fig. \ref{fig:grid} depicted as bipartite graphs}
\label{fig:bipartite}
\end{figure}

In general, if we consider a torus with two possible loops of permutations $\sigma_{x}, \sigma_{y}\neq Id$ (Fig. \ref{fig:lines}), we obtain games for which classical values do not depend on a particular type of non-trivial permutation, but solely on its presence (see Table \ref{Tabel1}). For such games, an upper quantum bound calculated based on first level of the NPA hierarchy, and lower than 1 due to periodic boundary conditions, coincide with an analytic quantum bound introduced in \cite{Ramanathan2016}. However, for non-periodic boundaries the latter bound exceeds 1. Classical values, as well as both bounds on quantum values, show an additive behavior: introduction of a new path connecting the boundaries leads to their decrease by a value that does not depend on presence of other paths of permutations. For example, if by $F(n,x,y)=C(n,x,y),Q^{\uparrow}(n,x,y),Q^{\downarrow}(n,x,y)$ we denote a function representing a classical value of a game or one of its quantum bounds, specified in Table \ref{Tabel1}, we have $1-F(n,1,1)=1-F(n,1,0)+1-F(n,0,1)$, which does not depend on the size of the lattice.

\begin{widetext}

\begin{table}

\small
\begin{tabular}{|c|c|c|c|c|c|c|c|c|c|}
\hline
\multirow{3}{*}{n} & \multirow{3}{*}{\rotatebox{90}{type}} & \multicolumn{8}{c|}{x,y} \\ \cline{3-10} 
   &        & \multicolumn{2}{c|}{1,0} &\multicolumn{2}{c|}{0,1}& \multicolumn{2}{c|}{1,1}& \multicolumn{2}{c|}{1,2}\\ \cline{3-10}
   &        & $C$ & $Q$ &$C$ & $Q$ &$C$ & $Q$ &$C$ & $Q$  \\\cline{1-10}
\multirow{4}{*}{2} & \multirow{2}{*}{a)} & \multirow{2}{*}{0.875} & 0.926666/ & \multirow{2}{*}{0.875} &  0.926666/ & \multirow{2}{*}{0.75} & 0.853553/& - & - \\ 
                              &                  &  & 0.926666 &  &  0.926666 &  & 0.853553 & - & -  \\ \cline{2-10}   
                               &\multirow{2}{*}{b)} & \multirow{2}{*}{0.875} & 0.926666/ & \multirow{2}{*}{0.9375} &  0.980970/ & \multirow{2}{*}{0.8125} & 0.907747/& - & -  \\ 
                              &                  &  & 0.926666 &  &  0.980970 &  & 0.907747 & - & -   \\ \cline{1-10}   
\multirow{4}{*}{3} &\multirow{2}{*}{a)} & \multirow{2}{*}{0.875} & 0.915578/ & \multirow{2}{*}{0.875} &  0.915578/ & \multirow{2}{*}{0.75} & 0.831812/& \multirow{2}{*}{0.75} & 0.833062
/   \\ 
                              &                  &  & 0.955342 &  &  0.955342 &  & 0.910684 &  & 0.910684  \\ \cline{2-10}   
                               &\multirow{2}{*}{b)} & \multirow{2}{*}{0.875} & 0.915357/ & \multirow{2}{*}{0.9375} &  0.977439
/ & \multirow{2}{*}{0.8125} & 0.893144
/& \multirow{2}{*}{0.8125} & 0.891906
/   \\ 
                              &                  &  & 0.955342 &  &  0.988642 &  & 0.943984 &  & 0.943984   \\ \cline{1-10}   
                              \hline
\multirow{3}{*}{n} & \multirow{3}{*}{\rotatebox{90}{type}} & \multicolumn{8}{c|}{x,y} \\ \cline{3-10} 
   &        & \multicolumn{2}{c|}{2,1} &\multicolumn{2}{c|}{0,2}& \multicolumn{2}{c|}{2,0}& \multicolumn{2}{c|}{2,2}\\ \cline{3-10}
   &        & $C$ & $Q$ &$C$ & $Q$ &$C$ & $Q$ &$C$ & $Q$  \\\cline{1-10}
\multirow{4}{*}{2} & \multirow{2}{*}{a)} & \multirow{2}{*}{-} & - & \multirow{2}{*}{-} &  - & \multirow{2}{*}{-} & -& - & - \\ 
                              &                  &  & -&  &  - &  & -& - & -  \\ \cline{2-10}   
                               &\multirow{2}{*}{b)} & \multirow{2}{*}{-} & - & \multirow{2}{*}{-} &  - & \multirow{2}{*}{-} &-& - & -  \\ 
                              &                  &  & -&  & -&  & - & - & -   \\ \cline{1-10}   
\multirow{4}{*}{3} &\multirow{2}{*}{a)} & \multirow{2}{*}{0.75} & 0.832910
/ & \multirow{2}{*}{0.875} &  0.915578
/ & \multirow{2}{*}{0.875} & 0.915578
/& \multirow{2}{*}{0.75} & 0.833325
/   \\ 
                              &                  &  & 0.910684 &  &  0.955342 &  & 0.955342 &  & 0.910684
  \\ \cline{2-10}   
                               &\multirow{2}{*}{b)} & \multirow{2}{*}{0.8125} & 0.892333/ & \multirow{2}{*}{0.9375} &  0.977439
/ & \multirow{2}{*}{0.9375} & 0.915520
/& \multirow{2}{*}{0.8125} & 0.891906  

/   \\ 
                              &                  &  & 0.943984 &  & 0.988642 &  & 0.955342 &  & 0.943984   \\ \cline{1-10}   
                              
\end{tabular}

\caption{Classical (\textit{C}) and quantum (\textit{Q}) values of non-locality games from Fig. \ref{fig:lines}, for given $x,y$ defining permutations $\sigma_{x}$ and $\sigma_{y}$ forming loops connecting periodic boundaries. For quantum case, upper $Q^{\uparrow}$ and lower $Q^{\downarrow}$ bounds shown in a format $\frac{Q^{\downarrow/}}{Q^{\uparrow}}$; they coincide for $n=2$. Classical and quantum values do not depend on a type of a non-trivial permutation. Increasing lattice length does not affect classical and quantum values for games with loops of non-trivial permutations along this direction.}\label{Tabel1}. 
\end{table}
\end{widetext}

\begin{figure}[h!]
\includegraphics[scale=0.40]{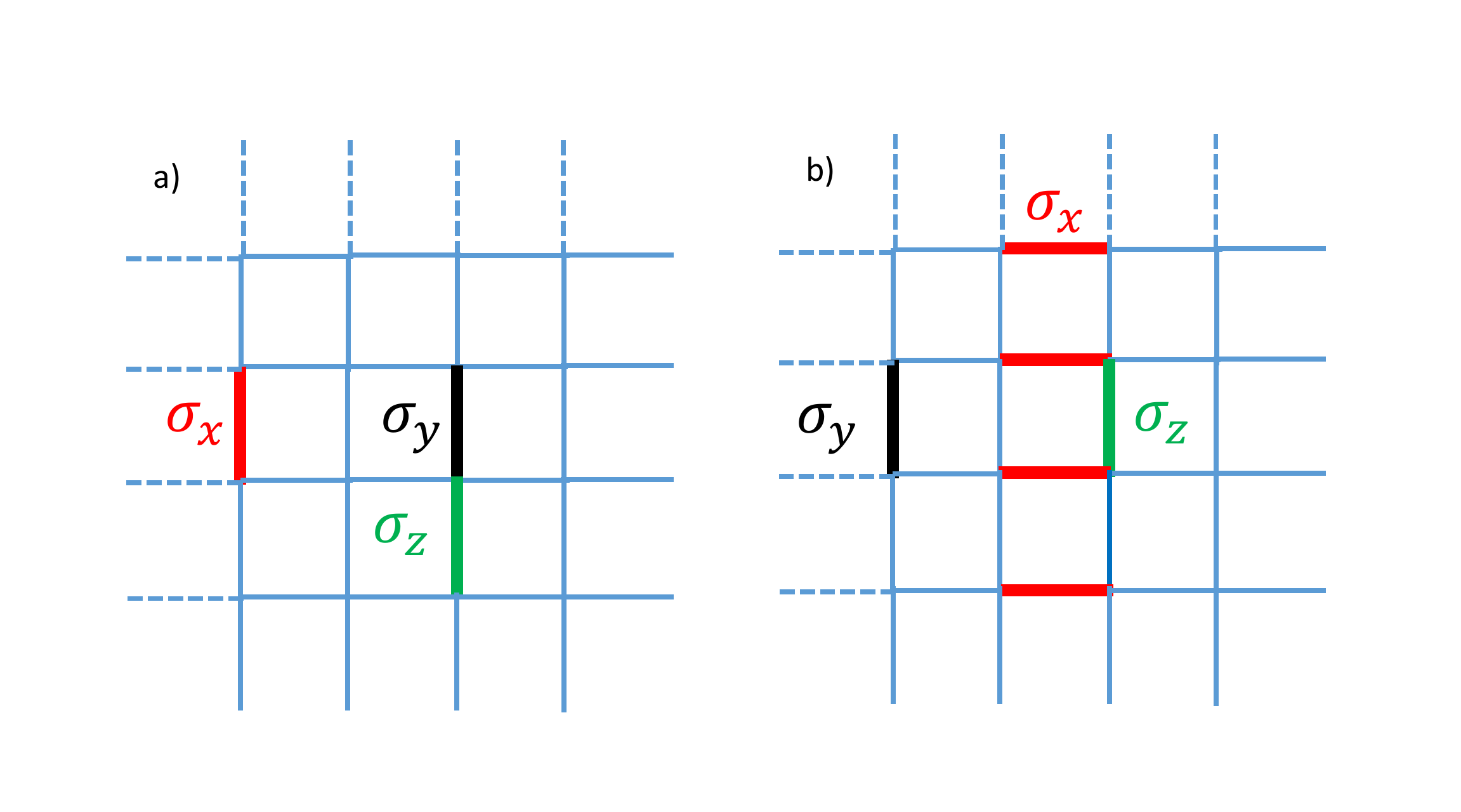}
\caption{Two non-locality games with periodic boundary conditions, with permutations $\sigma_{x}$ on red edges, $\sigma_{y}$ on black edges, $\sigma_{z}$ on green edges, and Identities on blue edges. Quantum values are given in Tables \ref{Tabel2} and \ref{Tabel3}.} 
\label{fig:line_edge}
\end{figure}

On the other hand, permutations that do not form closed loops influence quantum values in a non-additive manner (Fig. \ref{fig:line_edge}a). While both quantum bounds for $n=2$ coincide, their properties depend on position of respective permutations, e.g. for $n=2$ we have e.g. $1-Q(n=2,x=1,\textbf{y=0, z=1})\neq1-Q(n=2,x=1,\textbf{y=0, z=0}))+1-Q(n=2,x=0,\textbf{y=0, z=1})$, while  $1-Q(n=2,x=1,\textbf{y=1, z=0})\neq1-Q(n=2,x=1,\textbf{y=0, z=0}))+1-Q(n=2,x=0,\textbf{y=1, z=0})$ (see Table \ref{Tabel2}). Furthermore, we also see that $Q^{\downarrow}$ does not distinguish between games for $n=2$ and $n=3$ settings that have only one permutation present. Classical values for the game depend only on the number of non-trivial permutations present, which equals to $\beta_{C}$, so we have $C=\frac{32-\beta_{C}}{32}=0.90625, 0.9375, 0.96875$ for $3,2,1$ permutations, respectively. This indicates that going to a higher dimension may open a gap between classical and quantum values, as it is visible for cases of single permutations present. 

At the end, we show that the non-additive behavior of $Q^{\uparrow}$ is present for games with periodic boundary conditions even when it comes to adding a loop of permutations to a game already hosting a non-trivial permutation on one of the edges (see Fig. \ref{fig:line_edge}b ). Table \ref{Tabel3} shows bounds on quantum values, while classical values are independent of $n$ and equal to $0.96875, 0.875, 0.84375, 0.8125$ for a single permutation, loop connecting boundaries, loop with a single permutation and loop with 2 permutations present, respectively. Not only $1-Q^{\uparrow}(n,x=1,y=1,z=0)\neq1-Q^{\uparrow}(n,x=1,y=0,z=0)+1-Q^{\uparrow}(n,x=0,y=1,z=0)$, but a closed loop of permutations completely overshadows presence of single edge permutations for $n=2$.     

It is also visible from the above results that games that are equivalent according to Theorem \ref{Prop17} are characterized by the same bounds on quantum values. According to Theorem \ref{Prop17}, all games with permutation types, both local and in form of non-contractible loops, modified by the same permutation $\pi\in S_{d}$ and $\pi\notin L_{d}'$, are equivalent. 
\begin{widetext}

\begin{table}

\small
\begin{tabular}{|c|c|c|c|c|c|c|c|c|c|c|}
\hline
n &  \multicolumn{10}{c|}{} \\ \cline{1-11} 
\multirow{12}{*}{2} & \multicolumn{10}{c|}{x,y,z} \\ \cline{2-11} 
                               & 1,0,0 & 0,1,0 & 0,0,1 & 2,0,0 & 0,2,0 & 0,0,2 &0,1,1 & 0,2,1 & 0,1,2 & 0,2,2 \\\cline{2-11} 
 & 0.968750/& 0.968750/& 0.968750/& -/& -/& -/& 0.949843/& -& -& -\\
                               & 0.968750& 0.968750& 0.968750& - & - & -& 0.949843& -& -& -\\\cline{2-11}
						& \multicolumn{10}{c|}{x,y,z} \\ \cline{2-11} 
                               & 1,1,0 & 2,1,0 & 1,2,0 & 2,2,0 & 1,0,1 & 2,0,1 &1,0,2 & 2,0,2 & 1,1,1 & 2,2,2 \\\cline{2-11} 
 & 0.937842/& -& -& -& 0.937500/& -& -& -& 0.922388/& -\\
                               & 0.937842& -& -& -& 0.937500& -& -& -& 0.922388& -\\\cline{2-11}   
                               						& \multicolumn{10}{c|}{x,y,z} \\ \cline{2-11} 
                               & 1,1,2 & 1,2,1 & 2,1,1 & 2,2,1 & 2,1,2 & 1,2,2 & &  &  &  \\\cline{2-11} 
 & -& -& -& -& -& -& & & & \\
                               & -& -& -& -& -& -& & & & \\\cline{1-11}   

\multirow{12}{*}{3} & \multicolumn{10}{c|}{x,y,z} \\ \cline{2-11} 
                               & 1,0,0 & 0,1,0 & 0,0,1 & 2,0,0 & 0,2,0 & 0,0,2 &0,1,1 & 0,2,1 & 0,1,2 & 0,2,2 \\\cline{2-11} 
 & 0.968750/&  0.968750/&  0.968750/&  0.968750/& 0.968750/& 0.968750/& 0.942724/& 0.937500/& 0.937500/& 0.942724/\\
                               & 0.977378& 0.977378& 0.977378& 0.977378& 0.977378& 0.977378& 0.968772& 0.945183& 0.945183& 0.968772\\\cline{2-11}
						& \multicolumn{10}{c|}{x,y,z} \\ \cline{2-11} 
                               & 1,1,0 & 2,1,0 & 1,2,0 & 2,2,0 & 1,0,1 & 2,0,1 &1,0,2 & 2,0,2 & 1,1,1 & 2,2,2 \\\cline{2-11} 
 & 0.937500/& 0.937500/& 0.937500/& 0.937500/& 0.937500/& 0.937500/& 0.937500/& 0.937500/& 0.913291/& 0.913290/\\
                               & 0.958457& 0.951486&  0.951486&0.958457& 0.955486& 0.954088& 0.954088& 0.955486& 0.951016& 0.951016\\\cline{2-11}                                                               
                               & \multicolumn{10}{c|}{x,y,z} \\ \cline{2-11} 
                               & 1,1,2 & 1,2,1 & 2,1,1 & 2,2,1 & 2,1,2 & 1,2,2 & &  &  &  \\\cline{2-11} 
 & 0.906250/& 0.906250
/& 0.912307/& 0.906250
/& 0.906250/& 0.912307/& & & & \\
                               & 0.925112& 0.920754& 0.941841& 0.925112& 0.920754& 0.941841& & & & \\\cline{1-11}

\end{tabular}
\caption{Quantum values for a) of Fig. \ref{fig:line_edge}. For comparision: classical values are $0.96875, 0.9375, 0.90625$ for $1,2,3$ non-trivial permutations present, respectively.  }\label{Tabel2} 
\end{table}

\begin{table}

\small
\begin{tabular}{|c|c|c|c|c|c|c|c|c|c|c|}
\hline
n &  \multicolumn{10}{c|}{} \\ \cline{1-11} 
\multirow{12}{*}{2} & \multicolumn{10}{c|}{x,y,z} \\ \cline{2-11} 
                               & 1,0,0 & 0,1,0 & 0,0,1 & 2,0,0 & 0,2,0 & 0,0,2 &0,1,1 & 0,2,1 & 0,1,2 & 0,2,2 \\\cline{2-11} 
 & 0.926777/& 0.968750/& 0.968750/& -& -& -& 
0.937842
/& -/& -/& -/\\
                               & 0.926777& 0.968750& 0.968750& - & - & -& 0.937842& -& -& -\\\cline{2-11}
						& \multicolumn{10}{c|}{x,y,z} \\ \cline{2-11} 
                               & 1,1,0 & 2,1,0 & 1,2,0 & 2,2,0 & 1,0,1 & 2,0,1 &1,0,2 & 2,0,2 & 1,1,1 & 2,2,2 \\\cline{2-11} 
 & 0.896670/& -& -& -& 0.896670/& -& -& -& 0.866172/& -\\
                               & 0.896670& -& -& -& 0.896670& -& -& -& 0.866173& -\\\cline{2-11}   
                               						& \multicolumn{10}{c|}{x,y,z} \\ \cline{2-11} 
                               & 1,1,2 & 1,2,1 & 2,1,1 & 2,2,1 & 2,1,2 & 1,2,2 & &  &  &  \\\cline{2-11} 
 & -& -& -& -& -& -& & & & \\
                               & -& -& -& -& -& -& & & & \\\cline{1-11}   

\multirow{12}{*}{3} & \multicolumn{10}{c|}{x,y,z} \\ \cline{2-11} 
                               & 1,0,0 & 0,1,0 & 0,0,1 & 2,0,0 & 0,2,0 & 0,0,2 &0,1,1 & 0,2,1 & 0,1,2 & 0,2,2 \\\cline{2-11} 
 & 0.915578/& 0.968750 /&0.968750

 /&0.915577
  /&0.968750 /&0.968750 /&0.937500 /&0.937500 /&0.937500 /&0.937500 /\\
                              & 0.955342 & 0.977378& 0.977378 & 0.955342& 0.977378& 0.977378& 0.958457& 0.951486& 0.951486&0.958457 \\\cline{2-11}
						& \multicolumn{10}{c|}{x,y,z} \\ \cline{2-11} 
                               & 1,1,0 & 2,1,0 & 1,2,0 & 2,2,0 & 1,0,1 & 2,0,1 &1,0,2 & 2,0,2 & 1,1,1 & 2,2,2 \\\cline{2-11} 
 & 0.884399/ &0.884399/ &0.884399/ &0.884398
/ &0.884399/ &
0.884399/ &0.884399/ &0.884399/ &0.853247 / &0.853225 /\\
                               & 0.933402&  0.933402&  0.933402&  0.933402&0.933402& 0.933402& 0.933402& 0.933402 &0.914745 &0.914745\\\cline{2-11}                                                               
                               & \multicolumn{10}{c|}{x,y,z} \\ \cline{2-11} 
                               & 1,1,2 & 1,2,1 & 2,1,1 & 2,2,1 & 2,1,2 & 1,2,2 & &  &  &  \\\cline{2-11} 
 &0.853209 / &0.853210/ &0.853251/ &0.853210/ &0.853209/ &0.853248/ &  &&  & \\
                               & 0.908438& 0.908438&0.914745&0.908438&0.908438&0.914745&& & &\\\cline{1-11}

\end{tabular}
\caption{Quantum values for b) of Fig. \ref{fig:line_edge}. For comparison: classical values are $0.9687, 0.875, 0.84375, 0.8125$ for single permutation, loop connecting boundaries, loop with a single permutation and loop with 2 permutations present, respectively.}\label{Tabel3}. 
\end{table}
\end{widetext}

\section{Discussion and conclusions}\label{con}
Similarity between properties of the presented family of non-locality games and Kitaev error correction codes is not complete: as was shown for a $d=2$ case and a square lattice on a plane, a chain of anticorrelations in the dual lattice, which join two opposite boundaries, cannot be removed by local operations (i. e. by multiplying it by generators of stabilizer group), yet its analogue can be erased in the game setting by local relabelings of measurement outcomes.
 
Nevertheless, for $d=2$ and in a case of non-periodic boundary conditions of the lattice, the same algorithm can be used for calculating most probable chain of errors and classical value, in the quantum error correction code and non-locality games, respectively. Periodic boundary conditions do not impose any significant challenges, nor in game classification, nor in the computation of classical value of a game, and the problems can be addressed with modification of the algorithm used for non-periodic scenario.

The non-contractible loops of permutations appear to have  independent impacts on quantum values, as exact quantum values (for $d=2$) and both bounds on quantum values (for $d=3$) show additive behavior. This resembles commuting relation between logical operators acting in the codespace of the stabilizer error correction code, and enables us to conjecture that this is indeed property of a quantum set of correlations in the described setting. The non-additive behavior emerges only in presence of local features of the game -- defects, which, in the associated error correction picture, correspond to excitation of a state of the system out of the codespace, and happening due to local errors.

Due to game constructions on a square lattice, for a fixed surface its classical value cannot go to zero in the asymptotic limit of the number of questions asked, and it is lower bounded by $\frac{|T(G)|}{|E(G)|}$, with $|E(G)|$ being number of edges of a graph representing the game, and $|T(G)|$ number of edges belonging to its maximal spanning tree.  Nevertheless, by constructing games with relatively short maximal spanning trees, i.e. by utilizing surfaces with high genus and small minimal lengths of curves within each homotopy class, one could try to minimize classical value, and obtain high quantum/classical gap in high $d$, provided quantum values increase with number of possible outcomes.

\section*{Acknowledgments}
The authors would like to thank Justyna \L{}odyga and Waldemar K\l{}obus for useful discussions at the early stage of the project, and Anubhav Chaturvedi for helpful remarks and sharing his expertise in semidefinite programming. This work was supported by National Science Centre, Poland, grant OPUS 9. 2015/17/B/ST2/01945.

\section{Appendix}

In this paper we attempt to classify the labelings of a grid graph with permutations from the set $L_d'=\{\tilde{\sigma_i}:\tilde{\sigma_i}(x)=x+i\mod d\}$. We say that two labeled graphs $(G_1,K_1)$ and $(G_2,K_2)$ are\emph{equivalent} if one can be obtained from the other by means of the following operations.

\begin{enumerate}
\item Isomorphism between $G_1$ and $G_2,$

\item In a directed graph, replacing an arc $\overrightarrow{uv}$ labeled with $\pi$ with an arc $\overrightarrow{vu}$ labeled with $\pi^{-1},$

\item Switching operations $s(v,\sigma)$ for any vertex $v\in V_{G_1}$ and permutation $\sigma\in S_d,$ defined as follows. For every vertex $u\in N_{G}(v)$:

\begin{enumerate}
\item if $\overrightarrow{uv}\in E(G),$ we replace $K(\overrightarrow{uv})=\pi$ with $K'(\overrightarrow{uv})=\sigma\pi$,

\item if $\overrightarrow{vu}\in E(G),$ we replace $K(\overrightarrow{vu})=\pi$ with $K'(\overrightarrow{vu})=\pi\sigma^{-1}$.
\end{enumerate}
\end{enumerate}

However, in this paper we mostly consider an equivalence between different labelings of the same graph. We understand two labelings of a graph $G$ to be equivalent iff one can be obtained from the other by means of a series of switches $s(v,\sigma).$ It is clear that if the labelings $K_1, K_2$ of a graph $G$ are equivalent, then the labeled graphs $(G,K_1)$ and $(G,K_2)$ are also equivalent.

\begin{dfn}
A \textit{cell} of a lattice $G$ is a cycle which does not contain any other cycles within it. In the square lattice specifically, a cell is a cycle with four edges.
\\By a \textit{defect} we mean a cell which is a bad cycle, i.e. has no consistent vertex assignment.
\end{dfn}

\begin{dfn}
The \textit{defect class} $Cl(C)$ of a cycle $C=(v_0,e_1,v_1,...,e_k,v_0)$ in a labeled grid $(G,K)$ is the composition of all permutations assigned to the edges of the cycle, beginning from the starting vertex $v_0$ and moving in a counterclockwise direction, i.e. $cl(C,K)=K(e_1)K(e_2)...K(e_k)$.  
\end{dfn}

Conveniently, the defect class of a larger cycle can be calculated from the classes of the cells within the cycle.
In Proposition \ref{Tdodawanie_n2} we show that for $\tilde{\sigma}_{x}\in L_{d'}$, the number $x$ uniquely determines the defect class of the cell (in general though, the defect class of the cell should be identified with permutations itself). Similarly, the defect class of any cycle is defined by the composition of all permutations assigned to the edges of the cycle. By $D_x$ we denote the set of all defects of class $x$ in the labeled graph. The notion of the class of the cycle does not depend to the choice of a starting point $v_{0}$:

\begin{prop}
In a graph labeled with $L_{d'}$,  the defect class of a cycle does not depend on the starting point.
\end{prop}

\begin{proof}
Changing the starting point used in the definition of the defect class  will cause the change of the order of permutations $\tilde{\sigma}_{i}\in L_{d}'$ that $cl(C,K)$ is composed of, but this will not affect $cl(c,K)$, because all permutations in $L_{d}'$ commute with each other.
\end{proof}

\begin{prop}
\label{Tdodawanie_n2}
In a planar graph labeled with $L_d'$, the class of a cycle is the sum of the classes of all cells within the cycle.
\end{prop}
 
\begin{proof}
Let $C$ and $C'$ be two cycles in a planar graph which have two vertices $a$ and $b$ in common (see Fig. \ref{figura7}). We denote the compositions of all permutations on the $a-b$ path not belonging to $C'$ as $\pi_1$, on the $b-a$ path not belonging to $C$ as  $\pi_3$, and on the $b-a$ path belonging to both cycles as $\pi_2,$. 
Then we have $\Pi_C=\pi_2\pi_1$ and $\Pi_{C'}=\pi_3\pi_2^{-1}.$ The composition of permutations on the outer cycle $\hat{C}$ (containing both $C$ and $C'$) can be written as $\Pi_{\hat{C}}=\pi_3\pi_1=\pi_3\pi_2^{-1}\pi_2\pi_1=\Pi_C\Pi_{C'}.$ If the classes of $C$ and $C'$ are $x$ and $y$, respectively, we have $\tilde{\sigma}_y\tilde{\sigma}_x(a)=y+(x+a)=\tilde{\sigma}_{x+y}(a).$ Thus, if the cycle is formed on the boundary of two cycles of class $x$ and $y$ that share a common edge, its class will be $x+y$. 
\end{proof}

\begin{figure}[h!]
\includegraphics[scale=0.35]{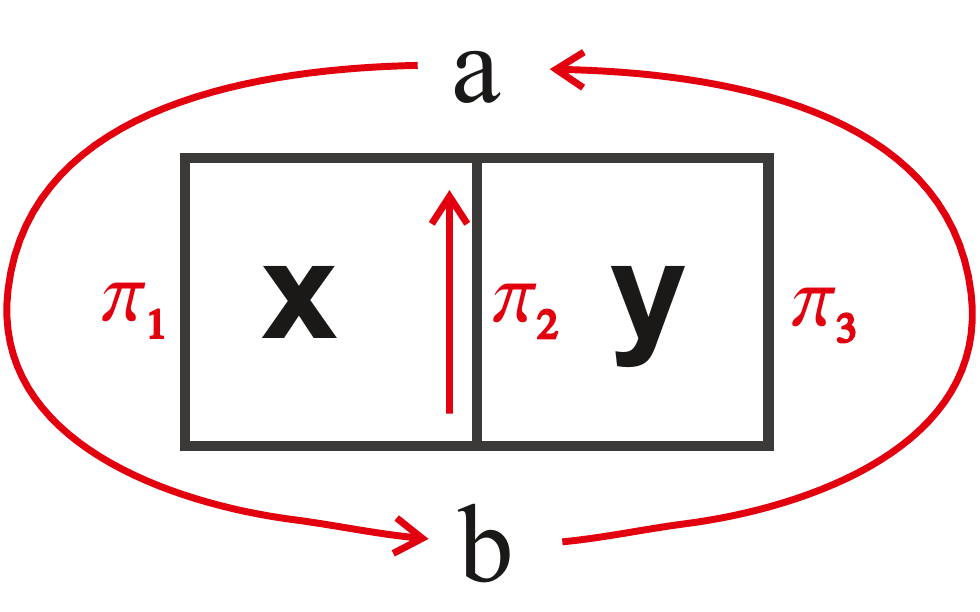}
\caption{\label{figura7} Addition of defects of class $x$ and $y$.}
\end{figure}

\begin{dfn}
By a (nontrivial) \textit{loop} we understand a series of edges in a grid such that the corresponding edges in the dual lattice form a cycle which is not contractible to a point. Typically we will refer to loops consisting of non-identity edges.
\end{dfn}

Any loop of length $k$ which exists in a grid with no defects defines a set of $k$ bad cycles in the graph, all of which have the same defect class. We will refer to the defect class of those cycles as the class $Cl(L)$ of the loop.

On a higher genus surface two labelings $K_1$ and $K_2$ with the same sets of defects of each class may be nonequivalent. In such cases the labeling $K_1-K_2$ contains no defect, but it has at least one loop. 

It would be convenient to define the equivalence classes based solely on the configuration of defects and loops within a labeling. However, unlike defects, loops may not actually be an inherent property of a given labeling. It is possible for $K_1-K_2$ to contain a loop even if no loop can be identified in either $K_1$ or $K_2$.

Which labeling in this case can be said to possess a loop? This is necessarily a matter of convention. For every possible configuration of defects one must choose one default labeling $D$ with the minimum number of non-identity edges. We will say that this labeling has no loops. Any labeling $K$ with the same configuration of defects as $D$ will be said to have a loop of class $x$ iff the labeling $K-D$ has such a loop.

Note that $K-D$ is a labeling with no defects. Thus, it is equivalent to a labeling in which all non-identity edges form unambigous loops.

\subsection{Classification}

\begin{lemma}
\label{L1}
If, by applying local permutations from $L_{d}'$ one can transform into each other two graphs labeled by permutations from $L_{d}'$, then these two labeled graphs have the same set of defects of each type, i.e. every cell has the same class in both graphs.
\end{lemma}

\begin{proof}

Let $\Pi=\tilde{\sigma}_{x}\in{L_{d}'}$ be the defect class of a selected cell. A switch $s(v,\tilde{\sigma}_y)$ changes the class of the cell $\Pi\rightarrow\tilde{\sigma}_{y}\tilde{\sigma}_{x}\tilde{\sigma}_{y}^{-1}=\tilde{\sigma}_{y}\tilde{\sigma}_{y}^{-1}\tilde{\sigma}_{x}=\Pi$ if $v$ is the starting vertex of $\Pi$, and $\Pi=\tilde{\sigma_{i}}\tilde{\sigma_{j}}\rightarrow\tilde{\sigma_{i}}\tilde{\sigma}_{y}\tilde{\sigma}_{y}^{-1}\tilde{\sigma_{j}}=\Pi$ otherwise.
\end{proof}

\begin{lemma}
\label{L2}
Applying switch $s(v,\tilde{\pi}_{y})$, where $\tilde{\pi_y}\in L_d$ to a vertex of a cell with defect class $x$ and labeled with $L_{d}'$ will change the defect class of the cell into $-x$, if the switch was applied to the starting vertex , and will not affect the defect class otherwise. 
\end{lemma}

\begin{proof}
Because for all $\tilde{\pi}_{x}\in L_{d}$ we have $\tilde{\pi}_{x}^{-1}=\tilde{\pi}_{x}$, the application of the switch to the starting vertex will change the cell's defect class in the following way $\Pi=\tilde{\sigma}_{x}\rightarrow\tilde{\sigma}_{y}\tilde{\sigma}_{x}\tilde{\sigma}_{y}=\tilde{\sigma}_{-x}$, because $\tilde{\pi}_y\tilde{\sigma}_x\tilde{\pi}_{y}(a) = y - (x + (y - a)) = -x + a.$ If the switch was applied to a different vertex, then we have $\Pi=\tilde{\sigma}_x=\tilde{\sigma}_i\tilde{\sigma}_j\rightarrow\tilde{\sigma}_i\tilde{\pi}_y\tilde{\pi}_y^{-1}\tilde{\sigma}_j=\tilde{\sigma}_x=\Pi$. 
\end{proof}

\begin{proof}[Proof of Theorem \ref{T:any_n}]

($\Rightarrow$)

If two labelings are equivalent, then one can be transformed into the other by switches. Assume without loss of generality that a permutation $\sigma_{1}\in L_d'$ on the edge $e=v_1v_2$ is transformed into $\sigma_{2}\in L_{d}'$ by switches $s(v_1,\eta)$ and $s(v_2,\pi)$: $\sigma_{2}=\pi\sigma_{1}\eta^{-1}$. Such a transformation changes the defect class of the cell $c_2$ with starting point $v_{2}$ from $cl(c_{2},K)$ into $\pi cl(c_{2},K) \pi^{-1}$, whereas the defect class of the cell $c_1$ with starting point $v_{1}$ is changed from $cl(c_{1},K)$ into $\eta cl(c_{1},K) \eta^{-1}$. Notice that $\eta=\sigma_{2}^{-1}\pi\sigma_{1}$, where $\sigma_{1},\sigma_{2}^{-1}\in L_{d}'$ and as such, they do not change the class of the cell. Now we see that the defect class of the cell $c_{1}$ changes into $\pi cl(c_{1},K) \pi^{-1}$. This implies $cl(c,L)=\pi cl(c,K)\pi^{-1}$ for all cells.

\item[($\Leftarrow$)]Let $K,L:E\mapsto L_d'$ be labelings such for every cell $c$, we have $cl(c,K)=\pi cl(c,L)\pi^{-1}$, where $\pi\in S_d$ is the same for all cells.
 When we start from the labeling $K$, and apply permutations $\pi$ on all vertices, we obtain an equivalent labeling $\pi(K)$ such that $cl(c,\pi(K))=cl(c,L)$ for every cell $c$.

 For every labeling with $L_{d'}$, we can obtain a canonical representation using only switches with permutations from $L_{d'}$. Such switches do not change the defect class of any cell. Furthermore, for a given spanning tree there is only one labeling with a given set of defects and all edges of the spanning tree labeled with identity. Thus, the labelings $\pi(K)$ and $L$ have a shared canonical representation, and are therefore equivalent.
\end{proof}

Note that the $(\Rightarrow)$ part of the proof does not depend on the graph being planar. Hence, we have the following corollary.

\begin{cor}
\label{cor:any_n}
Let $K$ and $L$ be two labelings of the same graph. If the labelings are equivalent, then there exists a permutation $\pi$ such that $cl(c,K)=\pi^{-1}cl(c,L)\pi$ for every cell $c$ of the graph. 
\end{cor}

\begin{proof}[Proof of Lemma \ref{L:tree+}]
It is obvious that all edges of $G-T^L$ belong to $G-L,$ as $T$ is a subgraph of $T^L.$ 

Now let us consider the dual lattice $G*.$ The graph $G*-T^L*$ has no vertices of degree $1$, as they were all absorbed into $T^L*$ during the enlarging procedure. It follows that every edge of $G*-T^L*$ belongs to a cycle and thus, every edge of $G-T^L$ belongs to a loop. 

Any contractible loop in $G$ divides its edge set into two disconnected components. Since $T^L$ is a connected graph, all loops contained in $G-T^L$ must be noncontractible.

The spanning tree $T$ contains no cycle. It follows that there is a loop in every homology class with no edges in $T.$ Since every cell belonging to such a loop has at least two edges not belonging to $T$, the loop dose not get absorbed into $T^L.$ Thus, $G-T^L$ contains a loop from every homology class.

Finally, two loops belonging to the same homology class can be continuously transformed into one another. But in this case these two loops would be separating the graph into two non-connected subgraphs. This, however, is impossible, since $T^L$ is connected. Therefore, for a given spanning tree $T$ and genus $g$ of the surface, there are 2g unique closed loops in the complement of $T^{L}$.

\end{proof}

In the proof of Theorem \ref{Prop17} we will use the fact that canonical form of a $L_{d}'$ game associated with a given spanning tree $T$ depends on permutations that belong to $T^{L}$, but not to $T$ (they are fixed, up to permutation $\pi$ applied to every vertex, by classes of defects), as well as by permutations that belong to the complement of $T^{L}$ (they are partially fixed, up to permutation $\pi$ applied to every vertex, by classes of defects, but also provide some additional characteristic of a game arising from periodic boundary conditions).

\begin{proof}[Proof of Theorem \ref{Prop17}]

First assume that there exists a permutation $\pi\in S_d$ such that $cl(c,K_1)=\pi cl(c,K_2)\pi^{-1}$ for every cell $c$ of the grid and $Cl(l,K_1)=\pi Cl(l,K_2)\pi^{-1}$ for every loop $l$. Let $T$ be a spanning tree of the graph. There is clearly only one way to label the enlarged tree $T^{L}$ such that all edges of $T$ are assigned the identity and the classes of all cells are preserved. The permutations assigned to the edges of $G-T^L$ must also be the same. If they were not, the classes of loops in the two labelings would be different. Thus $K_1$ and $\pi(K_2)$ have a shared canonical representation, which means that they are equivalent. Since $\pi(K_2)$ is equivalent to $K_2$, it follows that $K_1$ and $K_2$ are equivalent.

Now assume that $K_1$ and $K_2$ are equivalent. As shown before, $K_2$ can be obtained from $K_1$ by applying a switch $S(v,\sigma_v\pi)$ to each vertex, where $\sigma_v\in L_d'$ and $\pi\in S_d-L_d'$ is the same for all vertices.

As shown before, the switch $s(v,\sigma_v\pi)$ changes the defect class of a cycle from $\sigma$ to $\pi^{-1}\sigma\pi$ if $v$ is the starting point of the cycle and leaves it unchanged otherwise. This applies to contractible cycles on the surface, which arise from defects, as well as to noncontractible cycles, which give rise to nontrivial loops. 

Combined with the definition of the class of a loop, this shows that the set of switches which transforms $K_1$ into $K_2$ changes the classes of defects and loops in exactly the same way, namely  $cl(c,K_1)=\pi cl(c,K_2)\pi^{-1}$ for every cell $c$ and $Cl(l,K_1)=\pi Cl(l,K_2)\pi^{-1}$ for every loop $l$.

\end{proof}

\bibliographystyle{ieeetr}

\end{document}